\documentclass[letter,cleveref, autoref, thm-restate, 11pt]{article}

\title{Sparse Regular Expression Matching} 
\author{Philip Bille \\\texttt{phbi@dtu.dk} \and Inge Li G{\o}rtz \\\texttt{inge@dtu.dk}}

\date{}

\usepackage{amsmath, amssymb}
\usepackage{graphicx}
\usepackage[nocompress]{cite}
\usepackage{microtype}
\usepackage{enumerate}
\usepackage{comment}

\usepackage[dvipsnames,usenames]{xcolor}
\RequirePackage{hyperref}
\hypersetup{colorlinks=true, urlcolor=Blue, citecolor=Green, linkcolor=BrickRed, breaklinks, unicode}
\usepackage[margin=1in]{geometry}

\RequirePackage{hyperref}
\hypersetup{colorlinks=true, urlcolor=Blue, citecolor=Green, linkcolor=BrickRed, breaklinks, unicode}

\newcommand{\lef}{\ensuremath\mathsf{left}}
\newcommand{\rig}{\ensuremath\mathsf{right}}
\newcommand{\tops}{\ensuremath\mathsf{top}}
\newcommand{\bots}{\ensuremath\mathsf{bot}}
\newcommand{\succu}{\ensuremath\mathsf{succ}}
\newcommand{\predu}{\ensuremath\mathsf{pred}}

\newcommand{\first}{\ensuremath\mathsf{first}}
\newcommand{\firstextent}{\ensuremath\mathsf{firstextent}}

\newcommand{\lab}{\ensuremath\mathsf{label}}
\newcommand{\last}{\ensuremath\mathsf{last}}
\newcommand{\lastextent}{\ensuremath\mathsf{lastextent}}
\newcommand{\depth}{\ensuremath\mathsf{depth}}
\newcommand{\parent}{\ensuremath\mathsf{parent}}
\newcommand{\parentstar}{\ensuremath\mathsf{parent}^\ast}
\newcommand{\parentodot}{\ensuremath\mathsf{parent}^\odot}

\newcommand{\follow}{\ensuremath\mathsf{follow}}

\newcommand{\lca}{\ensuremath\mathsf{lca}}
\newcommand{\firstlabel}{\ensuremath\mathsf{firstlabel}}

\newcommand{\nex}{\ensuremath\mathsf{next}}

\newcommand{\Pos}{\ensuremath\mathsf{Pos}}
\newcommand{\RH}{\ensuremath\mathrm{RH}}

\newcommand{\AT}{\ensuremath A_{\mathsf{T}}}
\newcommand{\APOS}{\ensuremath A_{\mathsf{Pos}}}

\newcommand{\proofsubparagraph}[1]{\medskip
\noindent {\bf \color{darkgray}#1}. \hspace{0.5em}}

\usepackage{amsthm}

\newtheorem{theorem}{Theorem}

\newtheorem{lemma}[theorem]{Lemma}

\newtheorem{observation}[theorem]{Observation}
\newtheorem{claim}[theorem]{Claim}

\newenvironment{claimproof}[1]{\par\noindent\underline{Proof:}\space#1}{\hfill $\blacksquare$}

\begin{document}

\maketitle

\begin{abstract}
A regular expression specifies a set of strings formed by single characters combined with concatenation, union, and Kleene star operators. Given a regular expression $R$ and a string $Q$, the regular expression matching problem is to decide if $Q$ matches any of the strings specified by $R$. Regular expressions are a fundamental concept in formal languages and regular expression matching is a basic primitive for searching and processing data. A standard textbook solution [Thompson, CACM 1968] constructs and simulates a nondeterministic finite automaton, leading to an $O(nm)$ time algorithm, where $n$ is the length of $Q$ and $m$ is the length of $R$. Despite considerable research efforts only polylogarithmic improvements of this bound are known. Recently, conditional lower bounds provided evidence for this lack of progress when Backurs and Indyk [FOCS 2016] proved that, assuming the strong exponential time hypothesis (SETH), regular expression matching cannot be solved in $O((nm)^{1-\epsilon})$, for any constant $\epsilon > 0$. Hence, the complexity of regular expression matching is essentially settled in terms of $n$ and $m$. 

In this paper, we take a new approach and introduce a \emph{density} parameter, $\Delta$, that captures the amount of nondeterminism in the NFA simulation on $Q$. The density is at most $nm+1$ but can be significantly smaller. Our main result is a new algorithm that solves regular expression matching in $$O\left(\Delta \log \log \frac{nm}{\Delta} +n + m\right)$$ time. 

This essentially replaces $nm$ with $\Delta$ in the complexity of regular expression matching. We complement our upper bound by a matching conditional lower bound that proves that we cannot solve regular expression matching in time $O(\Delta^{1-\epsilon})$ for any constant $\epsilon > 0$ assuming SETH. 

The key technical contribution in the result is a new linear space representation of the classic position automaton that supports fast state-set transition computation in near-linear time in the size of the input and output state sets. To achieve this, we develop several new insights and techniques of independent interest, including new structural properties of the parse trees of regular expressions, a decomposition of state-set transitions based on parse trees, and a fast batched predecessor data structure.  
\end{abstract}

\thispagestyle{empty}
\setcounter{page}{0}

\newpage

\section{Introduction}
A regular expression $R$ specifies a set of strings formed by characters from an alphabet $\Sigma$ combined with concatenation ($\odot$), union ($\mid$), and Kleene star ($^\ast$) operators. For instance, $(a|(b\odot a))^\ast$ describes the set of strings of $a$s and $b$s such that every $b$ is followed by an $a$. Given a regular expression $R$ and string $Q$, the regular expression matching is to decide if $Q$ matches any of the strings specified by $R$. Regular expressions are a fundamental concept in formal language theory introduced by Kleene in the 1950'ties~\cite{Kleene1956} and regular expression matching is a basic tool in computer science for searching and processing text. Standard tools such as \texttt{grep} and \texttt{sed} provide direct support for regular expression matching in files, and the scripting language \texttt{perl}~\cite{Wall1994} is a full programming language designed to support regular expression matching easily. Regular expression matching appears in many large-scale data processing applications such as internet traffic analysis~\cite{JMR2007, YCDLK2006, KDYCT2006}, data mining~\cite{GRS1999}, data bases~\cite{LM2001, Murata2001}, computational biology~\cite{NR2003}, and human-computer interaction~\cite{KHDA2012}.

A classic textbook algorithm for regular expression matching, due to Thompson~\cite{Thomp1968} from 1968, constructs and simulates a nondeterministic finite automaton (NFA) $A$ in $O(nm)$ time, where $n$ is the length of $Q$ and $m$ is the number of character symbols in $R$. The simulation processes $Q$ from left to right and computes a sequence of sets of states $S_0, \ldots, S_n$ such that $S_i$ is the set of states in $A$ to which there is a path from the initial state that matches $Q[1..i]$. In 1985 Galil~\cite{Gal85} asked if a faster algorithm could be obtained. A sequence of results~\cite{Myers1992, BFC2008, Bille06, BT2009} improved the $O(nm)$ bound using tabulation or word-level parallelism leading to solutions using either $O(nm\frac{\log \log n}{\log^{1.5} n} + n + m)$~\cite{BT2009} or $O(nm\frac{\log w}{w} + n + m\log m)$ time~\cite{Bille06} time, where $w$ is the word length. Finally, Bille and Thorup~\cite{BT2010} gave an algorithm using $O(nk\frac{\log w}{w}+n+m\log k)$ time, where $k \leq m$ is the number of \emph{strings} appearing in the regular expression (see also \cite{DGGS2022, BM2011}).

The above solutions are based on the classic NFA simulation algorithm from Thompson's original algorithm~\cite{Thomp1968} and thus achieve the same $O(nm)$ time with some polylogarithmic factors shaved. Recently, conditional lower bounds have provided evidence for the lack of more significant progress. First, Backurs and Indyk~\cite{BI2016} showed in 2016 that we cannot solve regular expression matching in time $O((nm)^{1-\epsilon})$, for $\epsilon > 0$, assuming the strong exponential time hypothesis (SETH)~\cite{IR2001}. Then, in 2018 Abboud and Bringmann~\cite{AB2018} showed that we cannot solve the problem in time $O(nm/\log^{7+\epsilon} n)$, for $\epsilon > 0$, assuming the Formula SAT hypothesis~\cite{AB2018}. These results, together with Bringmann, Larsen, and Gr{\o}nlund~\cite{BGL2017} and Schepper~\cite{Schepper2020}, also studied subclasses of regular expression matching depending on the  structure of the operators in the expression, leading to a classification of the complexity of each such subclass. In summary, the complexity of regular expression matching is essentially settled in terms of $n$ and $m$. 

In this paper, we take a new approach and  introduce a \emph{density} parameter, $\Delta$, that captures the amount of nondeterminism in the NFA simulation on $Q$. The density is at most $nm+1$ but can be significantly smaller. Our main result is a new algorithm that solves regular expression matching in $$O\left(\Delta \log \log \frac{nm}{\Delta} +n + m\right)$$ time. This essentially replaces $nm$ with $\Delta$ in the complexity of regular expression matching. We complement our upper bound by a conditional lower bound that proves that we cannot solve regular expression matching in time $O(\Delta^{1-\epsilon})$ for any constant $\epsilon > 0$ assuming SETH.


\subsection{Sparse Regular Expression Matching}
Recall that the NFA simulation algorithm constructs a sequence $S_0, \ldots, S_n$ of state sets such that $S_i$ is the set of states in the automaton to which there is a path from the initial state that matches $Q[1..i]$. The goal of this paper is to explore the complexity of regular expression matching if these sets are \emph{sparse}. More precisely, let $(R, Q)$ be an instance of the regular expression matching and let $A$ be a finite automaton that accepts the set of strings defined by $R$, and let $S_0, \ldots, S_n$ be the sequence of sets of states in the simulation of $A$ on $Q$. We define the \emph{density} of $(R, Q)$ wrt. $A$ to be
$$
\Delta^A_{R,Q} = \sum_{i=0}^n |S_i|, 
$$
i.e., the density is the total size of the state sets in the simulation of $A$ on $Q$.  
\begin{figure}[t]
\centering  
\includegraphics[scale = 0.47]{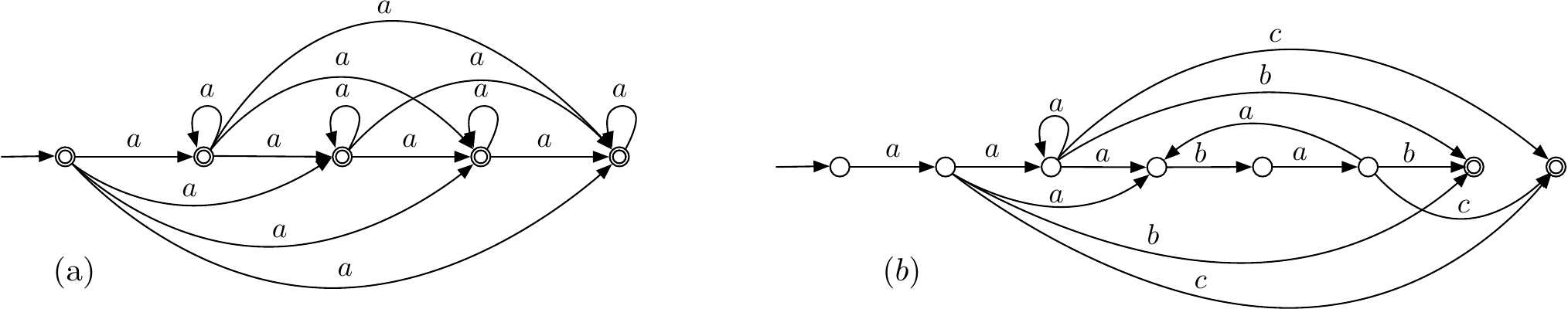}
  \caption{(a) The position automaton for the regular expression $a^\ast a^\ast a^\ast a^\ast$. (b) The position automaton for $a (a^\ast) (aba)^\ast (b|c)$.}
  \label{fig:automaton}	
\end{figure}
We will focus on density wrt. to the classic \emph{position automaton} (also known as \emph{Glushkov's automaton}), denoted $\APOS$,  proposed by Glushkov in 1960~\cite{Glushkov1960,Glushkov1961} and independently by McNaughton and Yamada~\cite{MY1960}. For most NFA constructions~\cite{GPW1998,AM2006,BHMMR2019,Thomp1968} (including several textbook constructions~\cite{AU1973,ASU1986,DDQ1978,LP1981,Martin1991, SS1988,Wood1987}), the density wrt. to $\APOS$ is a lower bound on the density wrt. to the other NFA construction. The key observation is that the set of states in $\APOS$ naturally corresponds to a \emph{subset} of the set of states in the other constructions. For instance, we can convert Thompson's NFA, $\AT$, into the corresponding position automaton, $\APOS$, by carefully contracting $\epsilon$-transitions~\cite{GPW1998,AM2006,BHMMR2019}. This implies $\Delta^{\APOS}_{R,Q} \leq \Delta^{\AT}_{R,Q}$ and thus if an algorithm is efficient in terms of $\Delta^{\APOS}_{R,Q}$ the same algorithm is also efficient in terms of $\Delta^{\AT}_{R,Q}$.  
Hence, for the rest of the paper, we define the density, denoted $\Delta_{R,Q}$, to be $\Delta^{\APOS}_{R,Q}$, and when $R$ and $Q$ are clear from the context we simply write $\Delta$. 

Intuitively, the density captures the amount of nondeterminism in the simulation of $\APOS$. At one extreme $\Delta = n+1$ when all of the $n+1$ state sets are singletons (assuming $R$ matches $Q$) and at the other extreme $\Delta = nm+1$ when all of the state sets, except the special $S_0$, consists of all states. The density can be significantly smaller than $nm$ in important practical scenarios. For instance, in internet traffic analysis a stream is matched against a large set of rules specified as a regular expression. Typically, most of these packets will only match a small subset of the rules implying a small density of the problem instance.

A related concept is \emph{deterministic regular expressions} (also known as \emph{$1$-unambiguous regular expressions)}. These are defined as regular expressions for which $\APOS$ is deterministic, that is, all state-set transitions on any singleton state set result in a singleton state set. Deterministic regular expressions are widely used in schema languages~\cite{BW1998, BPSMY2008, GST2012, MLMK2005} and have been extensively studied in complexity and automata theory~\cite{BW1998, BruggemannKlein1993, GM2017, PZC1996, MNS2010, CDLM2018, LBC2015}. Groz and Maneth~\cite{GM2017} showed how to solve the deterministic regular expression matching problem in $O(n\log \log m  + m)$ time. Note that if the regular expression is deterministic we always have that $\Delta \leq n+1$. 

\subsection{Sparse State-Set Transitions}
Given a set of $S$ of states and a character $\alpha$ a \emph{state-set transition}, denoted $\delta(S, \alpha)$, is the set of states reachable from $S$ via paths of transitions in the NFA that match $\alpha$ (for $\epsilon$-free NFAs the paths are always single transitions). We can implement the NFA simulation using $n$ state-set transitions by setting $S_0$ to be the initial state, and computing $S_{i} = \delta(S_{i-1}, Q[i])$ for $i = 1, \ldots, n$. In our scenario, we are interested in a compact representation of $\APOS$ that supports fast \emph{sparse state-set transitions}, i.e., a state-set computation that is efficient in terms of the sizes of the input set $|S|$ and the output set $|\delta(S,\alpha)|$. Since $\Delta$ is the total size of state sets in the simulation this implies an efficient algorithm for sparse regular expression matching.

Surprisingly, few results are known for this problem. If we store $\APOS$ explicitly we can compute $\delta(S, \alpha)$ by computing the union of the endpoints of transitions out of states in $S$ labeled $\alpha$. This leads to a data structure that uses $O(m^2)$ space and supports state-set transitions in $O(|S| |\delta(S, \alpha)|)$ time. Note that since endpoints of the transition may overlap (see Figure~\ref{fig:automaton}(a)) we may need to explore $\Omega(|S| |\delta(S, \alpha)|)$ transitions in general. A similar worst-case trade-off also holds for the many variants of the position automaton, see e.g.~\cite{Brzozowski1964, IY2003, Antimirov1996, Mirkin1966}. While $\epsilon$-free NFAs with fewer transitions are known~\cite{HSW2001,HMP2001, Geffert2003, Schnitger2006} these do not appear to translate to simulations for $\APOS$ nor do they improve the above time bound. 

Alternatively, we can store Thompson's automaton, $\AT$, and use the mapping of states mentioned above to convert state-set transitions on $\AT$ to state-set transitions on $\APOS$. Since $\AT$ is not an $\epsilon$-free automaton we can compute a state-set transition using a breadth-first search to explore all paths from $S$ that match $\alpha$. This uses $O(m)$ space and $O(m)$ time. However, it is easy to see that with this approach we may need to traverse large subgraphs of $\Omega(m)$ transitions labeled $\epsilon$ even if the sets $|S|$ or $\delta(S,\alpha)$ are sparse. Indeed, the efficient solutions in terms of $n$ and $m$ are based on improving state-set transitions in $\AT$ for the \emph{dense} case by polylogarithmic factors.

\subsection{Results}
Our main result is an efficient algorithm for sparse regular expression matching.  
\begin{theorem}\label{thm:regexmatching}
 	Given a regular expression $R$ with $m$ positions and a string $Q$ of length $n$, we can solve the regular expression matching problem in space $O(m)$  and time $$O\left(\Delta \log \log \frac{nm}{\Delta} + n + m\right).$$ 
\end{theorem}
Since the density $\Delta$ is at most $nm+1$, this essentially replaces $nm$ with $\Delta$ in the complexity of regular expression matching. As an immediate Corollary of Theorem~\ref{thm:regexmatching} we obtain a solution to deterministic regular expression matching using $O(n\log \log m + m)$ time and $O(m)$ space, thus matching the best known bound of Groz and Maneth~\cite{GM2017}. We complement Theorem~\ref{thm:regexmatching} with an essentially matching conditional lower bound. 
\begin{theorem}\label{thm:lowerbound}
For any $\Delta = n^{1+\gamma}$, for any constant $0 < \gamma \leq 1$, there exists no $O(\Delta^{1-\epsilon})$ time algorithm for regular expression matching for any constant $\epsilon>0$ assuming SETH. 
\end{theorem}

Theorem~\ref{thm:regexmatching} is based on a compact representation of the position automaton that supports efficient sparse state-set transitions. 
\begin{theorem}\label{thm:finiteautomaton}
	Given a regular expression $R$ with $m$ positions, we can represent the position automaton in $O(m)$ space and preprocessing time, such that given any set of states $S$ in sorted order and a character $\alpha$, we can compute the state-set transition $\delta(S,\alpha)$ in time $$O\left(|S|\log \log \frac{m}{|S|} + |\delta(S,\alpha)|\right).$$
	The output of the state-set transition is also reported in sorted order. 
\end{theorem}
The sorted order of $S$ and $\delta(S,\alpha)$ in  Theorem~\ref{thm:finiteautomaton} refers to the ordering of the corresponding positions in $R$ from left to right (without this condition the $\log \log (m/|S|)$ factor becomes $\log \log |S|)$). Theorem~\ref{thm:finiteautomaton} significantly improves the previous $O(|S| |\delta(S, \alpha)|)$ and $O(m)$ time bounds. Since any solution must use at least $\Omega(|S| + |\delta(S, \alpha)|)$ to read the input and write the output the bound is almost optimal.

\subsection{Techniques}

We develop several new insights and techniques of independent interest, including new structural properties of the parse trees of regular expressions, a novel decomposition of state-set transitions based on parse trees, and a fast batched predecessor data structure. 

We show how to decompose any state-set transition $\delta(S,\alpha)$  into a set of \emph{internal transitions} on a set of $O(|S| + |\delta(S,\alpha)|)$ \emph{transition nodes} of the parse tree of $R$. We have two types of internal transitions: one for $\odot$ and one for $\ast$. Intuitively, if $R(v) = R(u) \odot R(w)$ then the internal $\odot$-transition of $v$ wrt.\ $\alpha$ are all the states/positions in $R(w)$ reachable from a state/position in $R(u)$ using a transition labeled $\alpha$. The internal $\ast$-transitions are more complicated to describe, but both types of internal transitions are independent of the state set $S$. We show how to represent $R$ in linear space to efficiently compute internal transitions for any node $v$ and character~$\alpha$.

We identify the set of transition nodes for $\delta(S,\alpha)$ by first computing a compact representation of a \emph{transition tree}, which encodes all paths in $R$ containing transition nodes in $O(|S|)$ space. Then, we find the set of transition nodes using this tree. The key challenge is that even though the  representation of the transition tree is small the tree itself can be significantly larger and contain many nodes  that are irrelevant for the character $\alpha$ and/or irrelevant for the  state set $S$. Using the structural properties of the parse tree we show how to overcome the challenges and efficiently find the set of transition nodes for any $S$ and $\alpha$. Computing the internal transitions of all transition nodes could take too long, as the output of these overlap and we could end up using $\Omega(|S||\delta(S,\alpha)|)$ time. However, we prove that these output sets form a laminar family and show how to divide the computations of the internal transitions into computations on a bounded number of non-overlapping intervals. 

In combination, the above techniques lead to an $O(m)$ space representation of $R$ that supports state-set transitions in $O(|S|\log \log m + |\delta(S,\alpha)|)$ time. The bottleneck here is computing a $O(|S|)$ predecessor queries in $O(\log \log m)$ time. We present a simple two-level data structure that solves this batched predecessor problem in $O(|S|\log \log \frac{m}{|S|})$ time while maintaining linear space leading to our final structure. Using this solution for sparse-set transitions to implement the NFA simulation implies our main result for sparse regular expression matching of Theorem~\ref{thm:regexmatching}.

The lower bound follows from a reduction from the orthogonal vectors problem (OVP). We prove that given  $\Delta =n^{1 + \gamma}$, for any constant $0 < \gamma \leq 1$,  we can construct an instance of regular expression matching such that the existence an $O(\Delta^{1-\epsilon})$ algorithm for regular expression matching violates SETH.
The reduction is based on the reduction by Backurs and Indyk~\cite{BI2016} and is a fairly straightforward generalization of their lower bound.

\subsection{Related Work}
Another NFA construction, by Chang and Paige~\cite{CP1997}, considered compact representations of $\APOS$ that support efficiently implementing NFA to DFA conversion by subset construction. They presented a linear space representation that supports efficiently computing the set of states $S'$ reachable via \emph{any} character from a state-set $S$ in time $O(|S| + |S'|)$. Since $S'$ can be much larger than $\delta(S, \alpha)$ this does not imply an efficient sparse state-set transition.  

Some measures of nondeterminism of NFAs have been studied in automata theory, e.g., width, ambiguity, string tree width, string path width, and cycle height~\cite{KS2022, KM2019, LR2013, HS2002, KS2019}. These focus on the complexity of computing measures of the nondeterminism of a given NFA. In contrast, we study the complexity of regular expression matching in terms of the nondeterminism of a simulation on a given NFA and input string. 

As mentioned, Bille and Thorup~\cite{BT2010} considered the number of strings $k \leq m$ in the regular expression as a parameter for regular expression matching. They gave an algorithm using $O(nk\frac{\log w}{w}+n+m\log k) = O(nk + m\log k)$ time. It is straightforward to construct instances of regular expression matching (for a matching regular expression) such that either $nk = \Theta(nm)$ and $\Delta = \Theta(n)$ or $nk = \Theta(n)$ and $\Delta = \Theta(nm)$ hence this result is incomparable to ours. Cotumaccio, D’Agostino, Policriti, and Prezza~\cite{CDPP23} studied the indexing version of regular expressions, where the goal is to preprocess a regular expression in order to allow for fast matching given a query string. They considered the co-lexicographic width of an automaton. Applying their construction in the matching setting gives an algorithm that runs in $O(m^2 + np^2\log(p\cdot \sigma))$, where $p$ is the width of the co-lexicographic order and $m^2$ comes from the preprocessing of the automaton.

Several papers have studied the related problem of string matching in labeled graphs. For example, Rizzo, Tomescu, and Policriti~\cite{RTP22} studied the problem of matching a pattern on a labeled graph parameterized wrt the size of the labeled direct product graph, and Nellore, Nguyen, and Thompson~\cite{NNT21}  studied string matching in graphs parameterized by the size of the powerset automaton.

Finally, we note that sparsity is a well-studied phenomenon in a wide range of areas in computer science. In particular, sparsity has been extensively studied for other classic pattern matching problems, see, e.g.,~\cite{HS1977,EGGI1992a,EGGI1992, WL1983, Wilbur1984, AG1987}. 

\subsection{Outline}

We review regular expressions and automata in Section~\ref{sec:reg-and-automata} and the parse tree view of regular expressions in Section~\ref{sec:regextrees}. We introduce internal transitions, state-set decompositions, and transition trees in Section~\ref{sec:transition-nodes} and present our main algorithm for sparse state-set transitions in Section~\ref{sec:computingsparsetransitions}. In Section~\ref{sec:speedup}, we present the improved batched predecessor data structure. We use this to obtain the final result for sparse state-set transitions of Theorem~\ref{thm:finiteautomaton} which we then use to obtain Theorem~\ref{thm:regexmatching}. Finally, we show the lower bound of Theorem~\ref{thm:lowerbound} in Section~\ref{sec:lowerbound}.  

\section{Regular Expressions and Automata}\label{sec:reg-and-automata}
We briefly review the classical concepts used in the paper. For more details, see, e.g., Aho et al.~\cite{ASU1986}. 

\paragraph{Regular Expressions}
We consider the set of non-empty 
regular expressions over an alphabet $\Sigma$, defined recursively as
follows. If $\alpha \in \Sigma \cup \{\epsilon\}$ then $\alpha$ is a regular expression, and if $S$ and $T$ are regular expressions then so is the \emph{concatenation}, $(S)\odot(T)$, the \emph{union}, $(S)|(T)$, and the \emph{star}, $(S)^*$. We often omit the concatenation $\odot$ when writing regular expressions. The \emph{language} $L(R)$ generated by a regular expression $R$ is defined as follows. If $\alpha \in \Sigma \cup \{\epsilon\}$, then $L(\alpha)$ is the set containing the single string $\alpha$.  If $S$ and $T$ are regular expressions, then $L(S \odot T) = L(S)\odot L(T)$, that is, any string formed by the concatenation of a string in $L(S)$ with a string in $L(T)$, $L(S)|L(T) = L(S) \cup L(T)$, and $L(S^*) = \bigcup_{i \geq 0} L(S)^i$, where $L(S)^0 = \{\epsilon\}$ and $L(S)^i = L(S)^{i-1} \odot L(S)$, for $i > 0$.

\paragraph{Finite Automata}
A \emph{finite automaton} is a tuple $A = (V, E, \Sigma, \Theta, \Phi)$, where $V$ is a set of nodes called \emph{states}, $E \subseteq (V \times V \times \Sigma \cup \{\epsilon\})$ is a set of directed edges between states called \emph{transitions} each labeled by a character from $\Sigma \cup \{\epsilon\}$, $\Theta \subseteq V$ is a set of \emph{start states}, and $\Phi \subseteq V$ is a set \emph{accepting states}. In short, $A$ is an edge-labeled directed graph with designated subsets of start and accepting nodes. $A$ is a \emph{deterministic finite automaton} (DFA) if $A$ does not contain any $\epsilon$-transitions,  all outgoing transitions of any state have different labels, and there is exactly one start state. Otherwise, $A$ is a \emph{nondeterministic finite automaton} (NFA).  

Given a string $Q$ and a path $p$ in $A$ we say that $p$ and $Q$ match if the concatenation of the labels on the transitions in $p$ is $Q$. Given a state $s$ in $A$ and a character $\alpha$ we define the \emph{state-set transition} $\delta_A(s, \alpha)$ to be the set of states reachable from $s$ through paths matching $\alpha$ (note that the paths may include transitions labeled $\epsilon$). For a set of states $S$ we define $\delta_A(S,\alpha) = \bigcup_{s\in S} \delta_A(s,\alpha)$. We say that $A$ \emph{accepts} a string $Q$ if there is a path from a state in $\Theta$ to a state in $\Phi$ that matches $Q$. Otherwise, $A$ \emph{rejects} $Q$. We can use a sequence of state-set transitions to test acceptance of a string $Q$ of length $n$ by computing a sequence of state-sets $S_0, \ldots, S_n$, given by $S_0 = \delta_A(\Theta, \epsilon)$ and $S_i = \delta_A(S_{i-1}, Q[i])$, $i=1, \ldots, n$. We have that $\Phi \cap S_n \neq \emptyset$ iff $A$ accepts $Q$.

\paragraph{The Position Automaton} 
Given a regular expression $R$, we can construct an NFA accepting precisely the strings in $L(R)$ by several  classic methods~\cite{MY1960, Glushkov1961, Thomp1968}. In particular, Glushkov gave an important construction called the \emph{position automaton} or \emph{Glushkov automaton}. The position automaton is an $\epsilon$-free NFA consisting of only $m+1$ states and $O(m^2)$ transitions (See Figure~\ref{fig:automaton}). Each state except the start state corresponds to a position. Intuitively, each state-set in a state-set simulation is the set of positions in $R$ that correspond to a match of a prefix of $Q$. 

We review the details of the position automaton in the following. Let $R$ be a regular expression with $m$ character symbols from an alphabet $\Sigma$. The \emph{position} of a character in $R$ is the index of the character in the left-to-right order among the characters in $R$. The set of positions in $R$, denoted $\Pos(R)$, is the set $\{1,\ldots, m\}$. The \emph{label} of a position $p$, denoted $\lab(p)$, is the character at position $p$. The subset of positions labeled $\alpha$ is denoted $\Pos_\alpha(R)$. When $R$ is clear from the context we abbreviate $\Pos(R)$ to phb $\Pos$.

The \emph{marked regular expression} of $R$, denoted $\overline{R}$, is obtained from $R$ by subscripting each character in $R$ with its position. Similarly, the \emph{marked alphabet}, denoted $\overline{\Sigma}$, is obtained from $\Sigma$ by adding subscripts. The marked regular expression $\overline{R}$ defines the language $L(\overline{R})$ over the marked alphabet $\overline{\Sigma}$. Note that $\Pos(\overline{R}) = \Pos(R)$.  Given a position $p$ we define $\overline{\lab}(p)$ to be the label of $p$ in $\overline{R}$. The $\first$ and $\last$ set of $R$ represent the positions that match the first and last character, respectively, in some string in $L(\overline{R})$. Given a position $p$, the $\follow$ set of $R$ and $p$ is the set of positions that can follow a position $p$ in $L(\overline{R})$. More precisely,  
\begin{align*}
    \first(R) &= \{p \in \Pos(R) \mid \exists s \in \overline{\Sigma}^*, \overline{\lab}(p) \odot s \in L(\overline{R}))\} \\
    \last(R) &= \{p \in \Pos(R) \mid \exists s \in \overline{\Sigma}^*, s \odot \overline{\lab}(p) \in L(\overline{R}))\} \\
    \follow(R, p) &= \{q \in \Pos(R) \mid \exists s, t \in \overline{\Sigma}^*, s \odot \overline{\lab}(p) \odot \overline{\lab}(q) \odot t \in L(\overline{R}))\}
\end{align*}
We then define the position automaton for $R$ as the NFA $A = (V, E, \{0\}, F)$, where 
\begin{align*}
   V &= \Pos(R) \cup \{0\}, \\
   E &= \{(0, q, \lab(q)) \mid q \in \first(R)\} \; \cup \; \bigcup_{p \in \Pos}\{(p,q, \lab(q)) \mid q \in \follow(R, p)\} \\
   F &= 
   \begin{cases}
   \{0\} \cup \last(R) & \text{if $\epsilon \in L(R)$,} \\
   \last(R) & \text{otherwise}
   \end{cases} 
\end{align*}

\section{Regular Expressions as Trees}\label{sec:regextrees}
Throughout the rest of the paper, let $R$ be a regular expression with $m$ positions and let $\delta$ denote the state-set transition function of the position automaton for $R$. For simplicity in the presentation, we will focus on implementing $\delta$ on the positions of $R$ and ignore the extra start state of the position automaton. The extra start state is straightforward to represent with additional linear space and is only needed in the initial step of state-set simulations. 

\begin{figure}[t]
\centering  
\includegraphics[scale = 0.4]{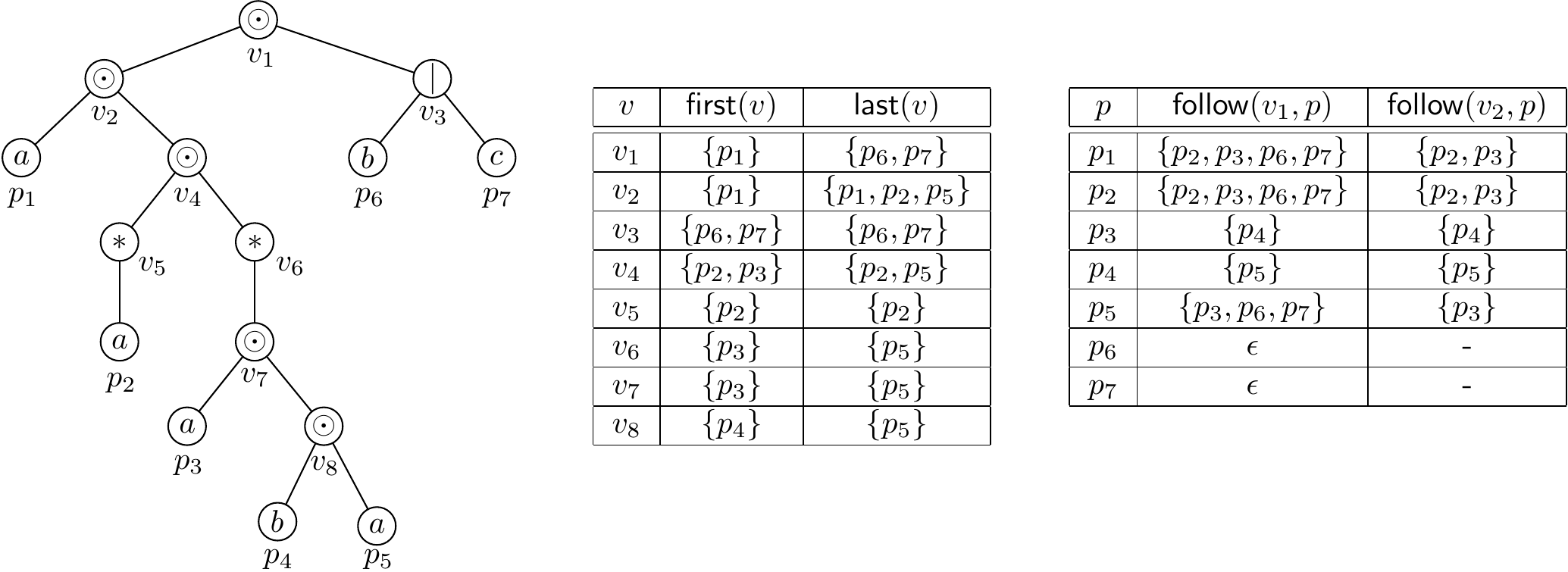}
  \caption{The parse tree for the expression $a (a^\ast)(aba)^\ast(b | c)$, the corresponding $\first$ and $\last$ sets, and the $\follow$ sets for $v_1$ and $v_2$. }
  \label{fig:exampleexpression}	
\end{figure}

We identify regular expressions by their parse trees (see Figure~\ref{fig:exampleexpression}). Note that the leaves in left-to-right order are the positions. We call the three types of internal nodes $\odot$-nodes, $\ast$-nodes, and $\mid$-nodes. For a $|$-node or $\odot$-node $v$ the left and right child are denoted $\lef(v)$ and $\rig(v)$, respectively, and for a $\ast$-node the single child is denoted $\lef(v)$. The \emph{depth} of a node $v$ in $R$ is the number of edges on the path from the root of $R$ to $v$. We denote the subtree (equivalently subexpression) rooted at a node $v$ by $R(v)$. If $u \in R(v)$ then $v$ is an ancestor of $u$, denoted $v \preceq u$, and if $u \in R(v)$ and $u\neq v$ then $v$ is a proper ancestor of $u$, denoted $v \prec u$. If $v$ is a (proper) ancestor of $u$ then $u$ is a (proper) descendant of $v$. A node $w$ is a common ancestor of $u$ and $v$ if it is an ancestor of both $u$ and $v$. The \emph{lowest common ancestor} of $u$ and $v$, $\lca(u,v)$, is the common ancestor of $u$ and $v$ of greatest depth.  The \emph{lowest star ancestor} of a node $v$, denoted $\parentstar(v)$, is the lowest ancestor of $v$ that is a $\ast$-node.

We extend the definition of labels to internal nodes. For each internal node $v$, the label of $v$, denoted $\lab(v)$, is a set of characters such that $\alpha \in \lab(v)$ iff $v = \lca(p,q)$ for some positions $p$ and $q$ both labeled $\alpha$. In Figure~\ref{fig:exampleexpression}, $\lab(v_2) = \{a\}$ since $v_2 = \lca(p_1, p_2)$ and $\lab(p_1)= \lab(p_2) =\{a\}$. Since the total number of internal nodes containing label $\alpha$ is $|\Pos_\alpha| - 1$ the total size of all labels is $O(m)$. For a node $v$ we extend our notation to define $\Pos(v)$, $\first(v)$, $\last(v)$, and $\follow(v,p)$ to denote the sets on the subexpression $R(v)$ (see Figure~\ref{fig:exampleexpression}). 

In our setting, we will often view the $\first$ and $\last$ sets from the perspective of a single position $p$ and consider the nodes for which $p$ appears in the corresponding $\first$ and $\last$ sets, respectively. Specifically, we define the \emph{first extent} and \emph{last extent} of a position $p$, respectively, to be the set of nodes in $R$ given by $\firstextent(p) =\{v \mid p \in \first(v)\}$ and $\lastextent(p) = \{v \mid p \in \last(v)\}$. Furthermore, for a set of positions $P$, we write $\firstextent(P) = \bigcup_{p \in P} \firstextent(p)$ and $\lastextent(P) = \bigcup_{p \in P} \lastextent(p)$. For instance, in Figure~\ref{fig:exampleexpression} we have $\firstextent(\{p_3, p_6\}) = \{p_3, v_7, v_6, v_4, p_6, v_3\}$. We define the first extent and last extent of an internal node $v$, to be the sets  $\firstextent(v) = \{ u \mid u \preceq v \textrm{ and } u \in \firstextent(\Pos(v))\}$ and 
$\lastextent(v) =\{u \mid u \preceq v \textrm{ and } u \in \lastextent(\Pos(v))\}$, respectively. 

The first sets and the last sets, respectively, form a laminar family. That is, for any two nodes in the parse tree, their first sets, respectively, last sets, are either disjoint or one is contained in the other. 
That implies that the set of nodes in $\firstextent(p)$, respectively, $\lastextent(p)$,  forms a path from position $p$ to an ancestor of $p$.

\begin{lemma}\label{lem:pathextent}
Let $p$ be a position in a regular expression $R$ and let $v$ and $u$ be nodes in $R$ such that $u \preceq v \preceq p$. If $u\in\firstextent(p)$, then $v\in\firstextent(p)$ and if $u\in\lastextent(p)$, then $v\in\lastextent(p)$.
\end{lemma}
\begin{proof}
 We have that $R(v)$ is a subexpression of $R(u)$ and $p$ is a position in $R(v)$. Then, if $p \in \first(u)$ then $p \in \first(v)$. Similarly, if $p \in \last(u)$ then $p \in \last(v)$.
\end{proof}

\section{Internal Transitions, State-Set Decompositions, and Transition Trees}\label{sec:transition-nodes}

We now introduce the main structural properties of state-set transitions that we need for our fast sparse state-set transition algorithm in Section~\ref{sec:computingsparsetransitions}. We first characterize state-set transitions in the position automaton in terms of $\firstextent$ and $\lastextent$ using the following important property.  
\begin{lemma}\label{lem:transition}
	Let $p \in \Pos$ and $q \in \Pos_\alpha$ and $v = \lca(p,q)$. Then, $q \in \delta(p, \alpha)$ iff either
	\begin{enumerate}[(i)]
		\item \label{odotcase} $v$ is a $\odot$-node, $\lef(v) \in \lastextent(p)$, and $\rig(v) \in  \firstextent(q)$, or
		\item \label{astcase} $\parentstar(v) \in \lastextent(p) \cap \firstextent(q)$.
	\end{enumerate}
\end{lemma}
Lemma~\ref{lem:transition} has appeared in various forms in earlier work~\cite{GM2017,CP1997,PZC1996}. For our purposes, we state it in terms of lowest common ancestors and first extents and last extents. Lemma~\ref{lem:transition} states that a position $q$ can only appear in $\delta(p,\alpha)$ through a $\odot$-node or a $\ast$-node. We write $q \in \delta^\odot(p,\alpha)$ if (\ref{odotcase}) is satisfied and $q \in \delta^\ast(p,\alpha)$ if (\ref{astcase}) is satisfied. 

\subsection{Internal Transitions}
Given an internal node $v$, an internal transition on $v$ and a character $\alpha$ will correspond to the conditions on $q$ and $v$ in Lemma~\ref{lem:transition} while ignoring the condition on $p$. In general, we will also specify a range of positions we are interested in. Formally, given an internal node $v$, a character $\alpha$, and positions $l$ and $r$, define the \emph{internal $\odot$-transition} and \emph{internal $\ast$-transition}, denoted $\delta^{\odot}(v, \alpha)$ and $\delta^{\ast}(v, \alpha)$, respectively, as follows.
\[
\begin{aligned}
\delta^{\odot}_{[l,r]}(v, \alpha) &= \{q \in \Pos_\alpha \mid \rig(v) \in \firstextent(q) \text{ and } q \in [l,r]  \} &&\text{if $v$ is a $\odot$-node} \\
\delta^{\ast}_{[l,r]}(v, \alpha) &= \{q \in \Pos_\alpha \cap \Pos(v) \mid \parentstar(v) \in \firstextent(q) \text{ and } q \in [l,r] 
\}  &&\text{if $v$ is a $\ast$-node}
\end{aligned} 
\]
When the range includes all positions we drop the subscript, that is, $\delta^{\odot}(v, \alpha) = \delta^{\odot}_{[1,m]}(v, \alpha)$ and $\delta^{\ast}(v, \alpha) = \delta^{\ast}_{[1,m]}(v, \alpha)$. For instance, in Figure~\ref{fig:exampleexpression} we have  $\delta^{\odot}(v_2, a) = \{p_2, p_3\}$, $\delta^{\odot}_{[3,5]}(v_2, a) = \{p_3\}$, $\delta^{\odot}(v_1, c) = \{p_7\}$ and $\delta^{\ast}(v_7, a) = \{p_3\}$.

\subsection{Transition Nodes}
Given a state-set transition $\delta(P,\alpha)$, the transitions nodes are a  set of nodes $N$ such that if we compute the union of internal transitions on $N$ we obtain $\delta(P, \alpha)$. Formally, we define the \emph{$\odot$-transition nodes} and \emph{$\ast$-transition nodes} of $\delta(P,\alpha)$, denoted $N^\odot(P,\alpha)$ and $N^\ast(P,\alpha)$, respectively, as 
\[
\begin{aligned}
  N^\odot(P,\alpha) &=   \{v \mid \text{$v$ is a $\odot$-node and $\lef(v)\in \lastextent(P)$ and $\rig(v) \in \firstextent(\Pos_\alpha)$} \} \\
  N^\ast(P,\alpha) &= \{v \mid \text{there exists $q \in \Pos_\alpha$ and $p \in P$ such that $v= \lca(p,q)$} \\ 
  &\phantom{=\{v \mid \text{ }} \quad \qquad \qquad \qquad \qquad \qquad \text{  and $\parentstar(v) \in \lastextent(p) \cap \firstextent(q)$}
  \}
\end{aligned}
\]
In combination, the set of \emph{transition nodes} is the union of the $\odot$-transition nodes and the $\ast$-transition nodes. 

\begin{lemma}\label{lem:transitiondecomp}
For any set of positions $P$ and a character $\alpha$,    
\begin{equation}\label{eq:transitiondecompositionx}
	\delta(P,\alpha) = \bigcup_{v\in N^\odot(P,\alpha)} \delta^{\odot}(v, \alpha) \qquad \cup \quad \bigcup_{v\in N^\ast(P,\alpha)} \delta^{\ast}(v, \alpha)\;.
\end{equation}
\end{lemma}
\begin{proof} 
 Let $\RH$ denote the right handside of (\ref{eq:transitiondecompositionx}).

We first show that $\delta(P, \alpha) \subseteq \RH$. Let $p \in P$ and $q \in \Pos$ be positions with $v = \lca(p,q)$ such that $q \in \delta(p,\alpha)$. Then, $q\in \Pos_\alpha$ and hence $p$, $q$, and $v$ satisfies either case (\ref{odotcase}) or (\ref{astcase}) in Lemma~\ref{lem:transition}. If  (\ref{odotcase}) is satisfied, $v$ is a $\odot$-node and $\lef(v) \in \lastextent(p) \subseteq \lastextent(P)$ and $\rig(v) \in  \firstextent(q) \subseteq \lastextent(\Pos_\alpha)$. By definition, $ v \in N^\odot(P,\alpha)$, and thus $q \in \delta^\odot(v,\alpha) \subseteq  \cup_{v\in N^\odot(P,\alpha)} \delta^{\odot}(v, \alpha)$. Similarly, if (\ref{astcase}) is satisfied, then $\parentstar(v) \in \lastextent(p) \cap \firstextent(q)$, and it follows $v\in N^\ast(P,\alpha)$. This implies that $q \in \delta^\ast(v,\alpha) \subseteq \cup_{v\in N^\ast(P,\alpha)} \delta^{\ast}(v, \alpha)$. 

To show $\RH \subseteq  \delta(P, \alpha)$ first suppose $q \in \cup_{v\in N^\odot(P,\alpha)} \delta^{\odot}(v, \alpha)$. Then, $q \in \Pos_\alpha$ and there is a $\odot$-node $v$ such that $\rig(v) \in \firstextent(q)$ and $\lef(v) \in \lastextent(P)$, which implies that $v = \lca(p,q)$ for some $p \in P$. By Lemma~\ref{lem:transition}(\ref{odotcase}) $q \in \delta(P,\alpha)$. 
If $q \in \cup_{v\in N^\ast(P,\alpha)} \delta^{\ast}(v, \alpha)$ then there exists a $v\in N^\ast(P,\alpha)$ such that $q \in \Pos(v)$ and $\parentstar(v) \in \firstextent(q)$. It follows from the definition of $N^\ast(P,\alpha)$ and By Lemma~\ref{lem:transition}(\ref{astcase}) that there exists a $p\in P$ and a $q'\in Pos_\alpha$ such that $v = \lca(p,q')$ and $\parentstar \in \lastextent(p) \cap \firstextent(q')$. Since both $p$ and $q$ are descendants of $v$  we have that $u = \lca(p,q)$ is a (not necessarily proper) descendant of $v$. It follows from Lemma~\ref{lem:pathextent} that $u \in \lastextent(p)$ and $\parentstar(u) \in \firstextent(q)$. Thus by Lemma~\ref{lem:transition}(\ref{astcase}) we have  $q \in \delta(P,\alpha)$. 
\end{proof}
We show that the total size of the two sets $N^\odot(P,\alpha)$ and $N^\ast(P,\alpha)$ is $O(|P| + |\delta(P,\alpha)|$). 

\begin{lemma}\label{lem:sizeNodotNast}
We have $|N^\odot(P,\alpha)| \leq |P| + |\delta^\odot(P,\alpha)|-1$  
and $|N^\ast(P,\alpha)| \leq |P| + |\delta^\ast(P,\alpha)|-1$. 
\end{lemma}
\begin{proof}
    By definition every node in $N^\odot(P,\alpha)$ is the lowest common ancestor of some position $p\in P$ and some position $q\in\delta^\odot(P,\alpha)$. The number of distinct pairwise lowest common ancestors of a subset of $\ell$ leaves in a tree cannot exceed $\ell-1$. Therefore, the number of lowest common ancestors between positions in $P$ and positions in $\delta^\odot(P,\alpha)$ can never be larger than $|P| + |\delta^\odot(P,\alpha)|-1$. The same argument holds for the number of nodes in $|N^\ast(P,\alpha)|$.
\end{proof}

\begin{figure}
    \centering
    \includegraphics[scale=0.33]{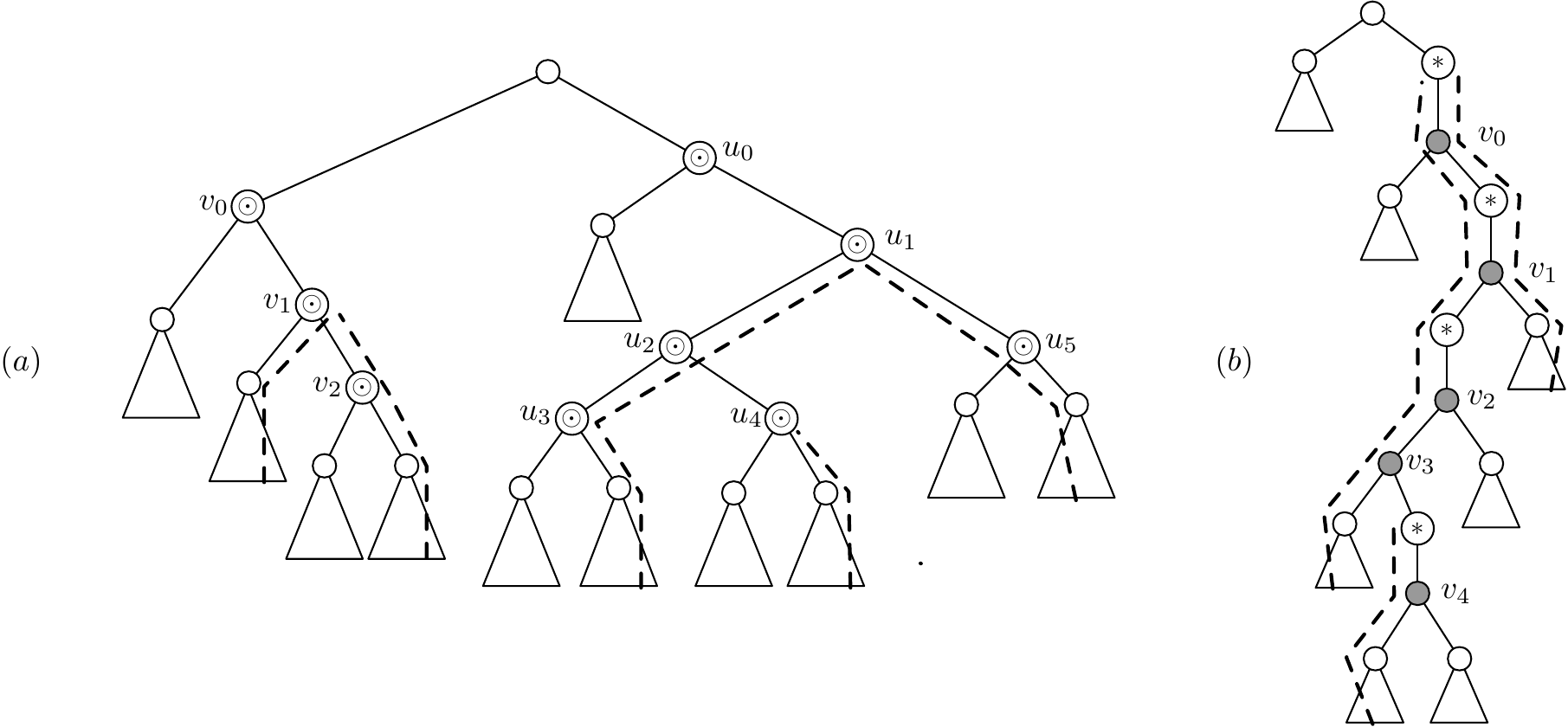}
    \caption{The dotted lines edges indicate the $\firstextent$ paths for the positions in $\Pos_\alpha$. (a) We have $\delta^\odot(v_2,\alpha) = \delta^\odot(v_1,\alpha)$ and they are contained in but not equal to $\delta^\odot(v_0,\alpha)$.  Both $\delta^\odot(u_3,\alpha)$ and  $\delta^\odot(u_5,\alpha)=\delta^\odot(u_1,\alpha)$  are contained in $\delta^\odot(u_0,\alpha)$.  Whereas,  $\delta^\odot(u_2,\alpha) \not\subseteq \delta^\odot(u_0,\alpha)$. Note, that we cannot dismiss e.g. $\delta^\odot(u_3,\alpha)$ at preprocessing time, since $\lef(u_0)$ might not be in $\lastextent(P)$.  (b) The grey nodes are lcas of a position in $P$ and a position in $\Pos_\alpha$.  The set $\delta^\ast(v_4,\alpha)$ is not included in any of the others. Nodes $v_2$ and $v_1$ has the same $\ast$-parent and thus $\delta^\ast(v_3,\alpha) \subseteq\delta^\ast(v_2,\alpha)$. Since  $\Pos_\alpha \cap \Pos(v_2) =\Pos_\alpha \cap \Pos(v_3)$ in the example then $\delta^\ast(v_3,\alpha) =\delta^\ast(v_2,\alpha)$. We also have $\delta^\ast(v_2,\alpha) \subset \delta^\ast(v_1,\alpha) = \delta^\ast(v_0,\alpha)$. We cannot dismiss any of the nodes at preprocessing time as  might be the case that only a subset (or none) of them is in $\lastextent(P)$.
    }
    \label{fig:odot-included}
\end{figure}


The internal transitions on the set of transition nodes are not disjoint and hence we cannot afford to compute internal transitions on each of the transition nodes explicitly. Fortunately, by Lemma~\ref{lem:pathextent}, the internal $\odot$-transitions (resp.\ $\ast$-transitions) of the nodes from $N^\odot(P,\alpha)$ (resp.\ $N^\ast(P,\alpha)$) form a laminar family. We use this to divide the computations of the internal transitions into computations on a bounded number of non-overlapping intervals. We implement this idea by compactly encoding all relevant transition nodes in the \emph{transition tree} defined in the following.  

\subsection{Transition Trees}\label{sec:transitiontree}
Let $P$ be a set of positions. Given a state-set transition $\delta(P,\alpha)$, we define the \emph{transition tree} $T$ as the subtree of $R$ induced by all nodes in $P$ and their ancestors (see Figure~\ref{fig:transitiontree}(a)). A \emph{segment} in $T$ is a path from a leaf or a branching node to (but not including) the nearest branching node above it (or to the root if no such branching node exists). The root node is its own segment (see Figure~\ref{fig:transitiontree}(b)). The  bottom node of a segment $s$ is denoted $\bots(s)$. Note that $\bots(s)$ is always a branching node, a leaf, or the root. Any branching node in $T$ is the lowest common ancestor of two nodes in $P$ and vice versa. Hence, we can compactly store $T$ in $O(|P|)$ space by storing $P$ and the branching nodes with pointers into $R$.

The following observation follows immediately from the fact that all nodes in $\lastextent(P)$ and their ancestors are contained in $T$.
\begin{observation}\label{obs:transitiontree} 
	Let $T$ be the transition tree for $P$ in $R$. If $v$ is a transition node for $\delta(P, \alpha)$ then $v$ is a node on a segment in $T$. 
\end{observation}
Thus it is enough to consider the nodes in the transition tree when computing the set of transition nodes.

\begin{figure}[t]
\centering  
\includegraphics[scale = 0.3]{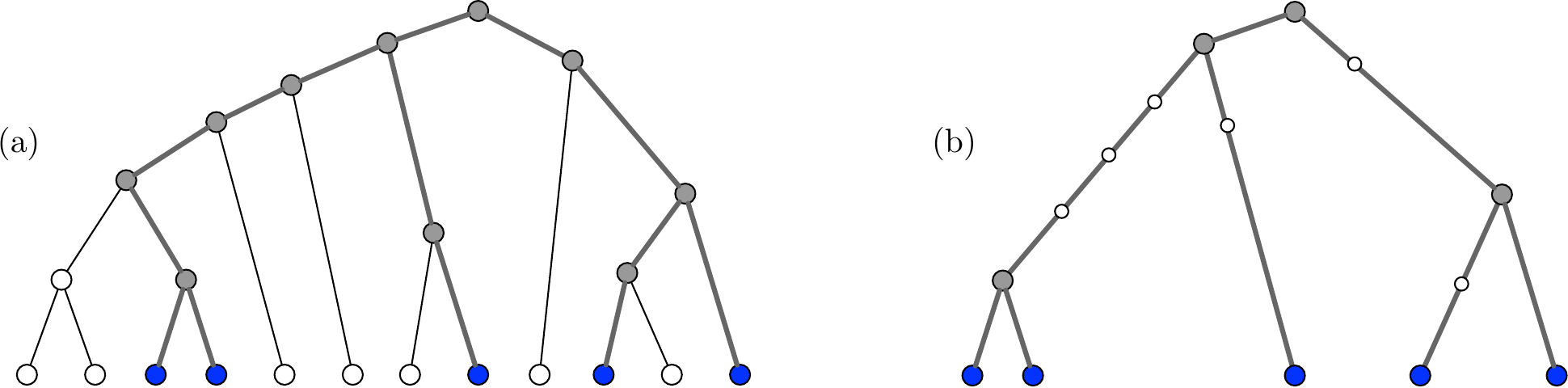}
  \caption{(a) The parse tree $R$. The blue leaves are the positions in $P$. (b) The transition tree $T$ of $P$. Branching nodes are grey and internal nodes on segments are white.}
  \label{fig:transitiontree}	
\end{figure}

\begin{observation}\label{obs:transitiontree_odotleft} 
Any $\odot$-node $u$ that is an internal node on a segment and  where $\lef(u) \in  \lastextent(P)$  must have its left child on the same segment. 
\end{observation}
To see why, observe that since $\lef(u) \in \lastextent(P)$, then $\lef(u)$ is an ancestor of a node in $P$, and thus by definition of $T$, $\lef(u)$ belongs to $T$. Since $u$ is an internal node on a segment $s$, only one of its children belongs to $T$. Therefore, $\lef(u)$ must also belong to segment $s$.

We will show in Section~\ref{sec:transitiontree-construct} that the transition tree $T$ and the key information we need in our algorithm can be computed in $O(|P|)$ time.


\section{Sparse State-Set Transitions}\label{sec:computingsparsetransitions}
We now present our $O(m)$ space data structure that supports computing a state-set transition $\delta(P, \alpha)$ in $O(|P|\log \log m + |\delta(|P|, \alpha)|)$ time.  We describe the data structure and analyze the space and preprocessing time in Section~\ref{sec:datastructure}.
The high-level idea of our sparse state-set algorithm is to identify the transition nodes for $\delta^\odot(P, \alpha)$ and $\delta^\ast(P, \alpha)$ using the transition tree. The state-set transitions for each set of nodes induce a partition of $\Pos_\alpha$ into nested intervals. We partition these intervals into non-overlapping intervals and then compute the internal transitions for each new interval. 
We then compute the union of these, which by the decomposition in Lemma~\ref{lem:transitiondecomp}, is precisely the set $\delta(P,\alpha)$. 
We describe how to construct the transition tree in Section~\ref{sec:transitiontree-construct}. In Section~\ref{sec:computetransitionnodes} we describe how to find the set transition nodes and how to construct the intervals. In Section~\ref{sec:computing-internal} we we show how to compute the internal transitions efficiently. Finally, in Section~\ref{sec:mainalg} we put everything together to get an algorithm for computing a state-set transition $\delta(P, \alpha)$ in $O(|P|\log \log m + |\delta(P, \alpha)|)$ time.

\subsection{Data Structure}\label{sec:datastructure}
We store the regular expression with labels on leaves and internal nodes together with the following components. 

\begin{itemize}
    \item For each node $v$ in $R$, we store the range of positions that are descendants of $v$, the depth of $v$, the highest node in $\lastextent(v)$, and the highest node in $\firstextent(v)$. We also store a pointer $\parentodot(v)$ to the lowest ancestor $u$ of $v$ such that $u$ is an $\odot$-node and $v\in T(\lef(u))$, and a pointer $\parentstar(v)$ to its lowest star ancestor.
	
 	\item At each branching node $v$ we store the position of the rightmost leaf in $\lef(v)$  and the position of the leftmost leaf in $\rig(v)$.
    \item Data structures for $R$ that supports \emph{lowest common ancestor queries} and \emph{first label queries}. Given a node $v$ and a character $\alpha$, a first label query, denoted $\firstlabel(v,\alpha)$, returns the lowest ancestor of $v$ whose label contains $\alpha$.
  
\end{itemize}
Furthermore, we store the following information for each $\alpha \in \Sigma$.
\begin{itemize}
    \item Arrays $A_\alpha$
	 and $D_\alpha$, where $A_\alpha[i]$ is the $i$th position labeled $\alpha$ in the left-to-right ordering of $\Pos_\alpha$, and $D_\alpha[i]$ is the depth of the highest node in $\firstextent(A_\alpha[i])$. 
	
    \item A data structure supporting predecessor and successor queries on the positions in $\Pos_\alpha$. That is, given any position $p$ in $R$ the predecessor (successor) query returns the position in $A_\alpha$ of the nearest position labeled $\alpha$ to the left (right) of $p$ including $p$ itself. For a branching node $v\in R$, we define the successor of $v$ in $\Pos_\alpha$ as the successor in $\Pos_\alpha$ of the leftmost leaf in $\rig(v)$. Note that this corresponds to the first position labeled $\alpha$ after $v$ in the order obtained by an inorder traversal of the nodes in $R$. Similarly, we define the predecessor of $v$ in $\Pos_\alpha$  as the predecessor  in $\Pos_\alpha$ of the rightmost leaf in $\lef(v)$. 
    
    \item A data structure on $D_\alpha$ that supports \emph{range minimum queries}. Given any pair of indices $l$ and $r$, the range minimum query on $D_\alpha$ returns a minimum value in the subarray $D_\alpha[l,r]$.
    
	\item For each node $v$ containing label $\alpha$:
	\begin{itemize}
		\item A pointer $\nex^\odot(v,\alpha)$ to the lowest proper ancestor $u$ of $v$ such that $v\in T(\lef(u))$,  $\alpha\in \lab(u)$, and $\delta^{\odot}(u,\alpha)$ is non-empty. If no such $u$ exists we store a null pointer to indicate this. 
		  \item A pointer $\nex^{\ast}(v, \alpha)$ to  the lowest proper ancestor $u$ of $v$ labeled  $\alpha$ such that there exists a $q\in\delta^\ast(u,\alpha)$ where $q \in T(\rig(u))$ if $v \in T(\lef(u))$ and $q\in T(\lef(u))$ if $v\in T(\rig(u))$. If no such $u$ exists we store a null pointer to indicate this.
        \item The range of positions in $A_\alpha$ that are descendants of $v$, of $\lef(v)$, and of $\rig(v)$, respectively.
	\end{itemize}
	
\end{itemize}

The idea of the $\nex^\odot$-pointers is that they form a chain of prospective nodes for $N^\odot$ with label $\alpha$. Any node $u$ from $N^\odot(P,\alpha)$ with label $\alpha$ on a segment $s$ has its left child on the path and $\delta^{\odot}(u,\alpha) \neq \emptyset$, so it is included in this chain. We show that  at most one node can be from $N^\odot(P,\alpha)$ on each segment that does not have label $\alpha$. Furthermore, given a segment $s$, the nodes from the chain on $s$ that belong to $N^\odot(P,\alpha)$ form a subchain starting from the lowest node of the chain that is on $s$ to (not including) the first node in the chain that is either  not in $\lastextent(P)$ or not on $s$.  Similarly, the $\nex^\ast$-pointers form a chain of prospective nodes for $N^\ast$ with label $\alpha$.

\paragraph{Space}
The regular expression and the labels use $O(m)$ space. The arrays $A_\alpha$ and $D_\alpha$,  $\alpha \in \Sigma$, use $O(m)$ space in total. We use linear space and linear preprocessing time data structures to support lowest common ancestors in constant time~\cite{HT1984,BFC2000}, first label queries in $O(\log \log m)$ time~\cite{Die89}, predecessor queries in $O(\log \log m)$ time~\cite{Willard1983,MN1990}, and range minimum queries in constant time~\cite{HT1984,BFC2000}. For each alphabet character, the total size of these data structures is linear in the number of leaves labeled with that character. Thus in total the space for these data structures is linear in $m$. The cited data structures for predecessor queries both use randomization, but since we only need a static structure it is straightforward and well-known how to obtain the same bound deterministically by combining deterministic dictionaries~\cite{HMP2001} with a simple two-level approach (see, e.g., Thorup~\cite{Thorup2003}). We store at most two $\nex$ pointers for each label in $R$ and a single pointer for each position using $O(m)$ space.  The remaining information uses $O(m)$ space. 

\paragraph{Preprocessing} We compute the range of positions that are descendants of $v$, the depth of $v$, the $\parentodot(v)$ and $\parentstar(v)$ pointers, and the positions of $\lef(v)$ and $\rig(v)$ using tree traversal in linear time. 
To compute the highest node for each node in $R$,  we first compute for each node  $u \in R$ if $\epsilon \in R(u)$ using a linear time bottom-up tree traversal. In top-down traversal we then compute  the highest node $H_L(v)$ in $\lastextent(v)$ for each node $v\in R$ using the following rules: If $v$ is the left child of an $\odot$-node and $\epsilon \not\in R(\rig(v))$ then $H_L(v) = v$. Otherwise, $H_L(v) = H_L(\parent(v))$. We compute the highest node in $\firstextent(v)$ similarly. We construct the arrays $A_\alpha$ and $D_\alpha$, for all $\alpha \in \Sigma$ in a single tree traversal.

To compute the remaining information we do the following for each $\alpha\in \Sigma$. Construct a tree $R_\alpha$ containing all nodes with label $\alpha$. To do this, we use $\lca$ queries on each consecutive pair of leaves in $A_\alpha$ from left to right. By keeping track of the depths of the nodes and checking if the newest node is an ancestor of the previous node it is straightforward to implement this in linear time. To compute the $\nex$-pointers we do a top-down traversal of $R_\alpha$. In each node $u$ we check if $\delta^\odot(u,\alpha)$ and $\delta^\ast(u,\alpha)$ are empty using range minimum queries on $D_\alpha$. To check if $\delta^\odot(u,\alpha)$ we do the range minimum query on the range of positions in $\Pos_\alpha$ in $\rig(v)$. If the depth returned is less than or equal to $\depth(\rig(v))$ then $\delta^\odot(u,\alpha)$ is non-empty.  For $\delta^\ast(u,\alpha)$ we do the query separately on the intervals corresponding to the left and the right child and we compare with the depth of $\parentstar(v)$. With this information, we can  compute the $\nex$-pointers during the traversal of $R_\alpha$ in constant time per node. The total size of all the $R_\alpha$ trees $O(m)$, since each tree has size $2|\Pos_\alpha| -1$. Thus the total time used for each $\alpha$ is $O(|\Pos_\alpha|)$.
 Hence, it follows that the total preprocessing time is $O(m)$.

\subsection{Constructing the Transition Tree}\label{sec:transitiontree-construct}

We say that a node $v$ is a \emph{$\odot$-live node} if $\lef(v) \in \lastextent(P)$ and $v$ is a $\odot$-node.
Note that any node in $N^\odot(P,\alpha)$ is a $\odot$-live node and a node in the transition tree $T$. It follows from Observation~\ref{obs:transitiontree_odotleft} that any $\odot$-live node that is an internal node on a segment has its left child on the segment. A segment $s$ in $T$ is called a \emph{$\ast$-segment} if $\bots(s)$ is not the root and  $\parentstar(\bots(s))\in \lastextent(P)$. 

We compute the compact representation of the transition tree as follows. Let $P$ be the set of leaves and repeatedly take the $\lca$ of adjacent nodes to form the internal nodes and the segments of $T$. Using a tree traversal on the compact transition tree we also compute for all branching nodes $v$ in $T$ the depth of the highest node in $\lastextent(P \cap \Pos(v))$, all $\odot$-live branching nodes in $T$, and all $\ast$-segments of $T$. Hence, we have the following result.

\begin{lemma}\label{lem:transitiontree-compute}
In $O(|P|)$ time we can compute
    the transition tree $T$ of $P$,
    for all branching nodes $v$ in $T$ the depth of the highest node in $\lastextent(P \cap \Pos(v))$,
    all $\odot$-live branching nodes in $T$, and
    all $\ast$-segments of $T$.
\end{lemma}

\begin{observation}
We can check in  constant time  if a node $v\in T$ is in $\lastextent(P)$ given the segment it is on.
\end{observation}

If $v$ is a leaf or a branching node, we already computed the information. Otherwise, $v$ is an internal node on a segment $s$. Then we compare the depth of $v$ with the depth $d$ of the highest node in $\lastextent(P \cap \Pos(\bots(s)))$. Now $v$ is in $\lastextent(P)$ if and only if the depth of $v$ is at least $d$.

\subsection{Computing Transitions Nodes and Intervals }\label{sec:computetransitionnodes}

We construct two sets of nodes $M^\odot$ and $M^\ast$ that consists of all nodes in $N^\odot(P,\alpha)$ and $N^\ast(P,\alpha)$, respectively, together with a constant number of other nodes per segment. We compute these sets for each segment using a depth-first traversal of the transition tree. We also construct sets $L^\odot$ and $L^\ast$, that partition $\Pos_\alpha$ into intervals in order to avoid recomputing overlapping internal transitions. We associate each interval with the lowest node from $M^\odot$ (respectively $M^\ast$) that can contain the positions in its internal transition (see Figure~\ref{fig:intervals}).

\begin{figure}[t]
\centering  
\includegraphics[scale = 0.35]{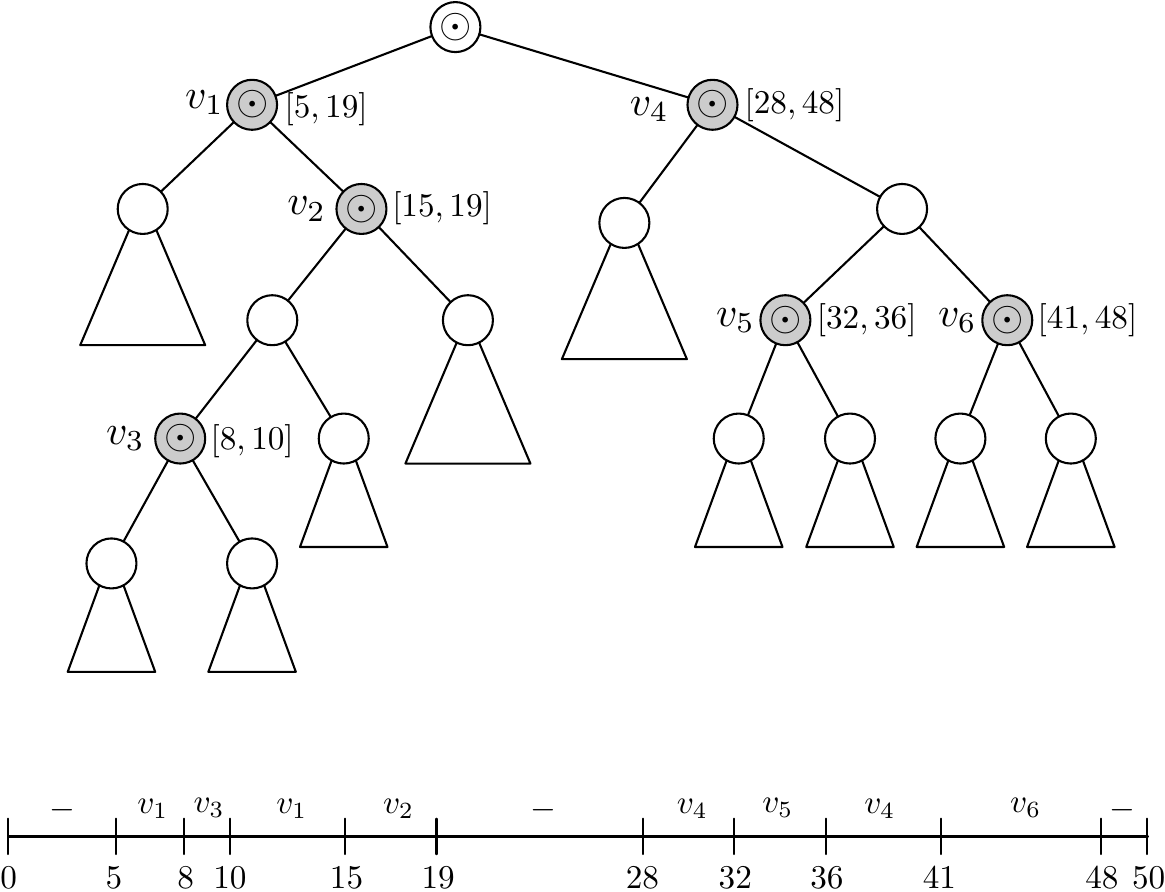}
  \caption{The nodes in $M^\odot$  are colored grey. The interval $[l_v,r_v]$ covered by $\rig(v)$ is written next to the nodes. The list $L^\odot=\{(5,7,v_1), \allowbreak (8,10,v_3),  \allowbreak (11,14,v_1),  \allowbreak (15,19,v_2),  \allowbreak (28,31,v_4), (32,36, v_5), (37,49,v_4), (41,48,v_6) \}$. In  the final algorithm, we compute and return the union of $\delta_{[5,7]}^\odot(v_1, \alpha)$, $\delta_{[8,10]}^\odot(v_3, \alpha)$, $\delta_{[11,14]}^\odot(v_1, \alpha), \ldots, \delta_{[41,48]}^\odot(v_6, \alpha)$.}
  \label{fig:intervals}	
\end{figure}

\subsubsection{Computing Transitions Nodes}
We first explain how to compute the transition nodes using the $\nex$-pointers and the transition tree. 

\paragraph{Computing $M^\odot$} We find for each segment $s$ all $\odot$-transition nodes for $\delta(P, \alpha)$ as follows. 
\begin{enumerate}
 
\item\label{alg:firstodot}  \emph{Find first relevant $\odot$-node on $s$: } Set $x = \parentodot(\bots(s))$. If $x$ is not on $s$ or $x$ is not an $\odot$-live node stop.  
 
\item\label{alg:firstodottrans} \emph{Find first $\odot$-transition node on $s$: } Compute the successor $q$ of $x$ in $\Pos_\alpha$, i.e., $q$ is the successor in $\Pos_\alpha$ of the leftmost leaf in $R(\rig(x))$. If no such $q$ exists stop,  otherwise set $v=\lca(x,q)$.  If $v$ is not on $s$ stop. 
 If $v$ is a $\odot$-live node not labeled $\alpha$  we add $v$ to $M^\odot$. Now compute $x = \firstlabel(v, \alpha)$. If $x$ is not a $\odot$-live node set $x = \nex^{\odot}(x,\alpha)$. 
 
\item\label{alg:nextodot} \emph{Repeatedly find next transition node by following $\nex^{\odot}$-pointers: } 
  We find all $\odot$-transition nodes by repeatedly following $\nex^{\odot}$-pointers from $x$ and adding the visited nodes to $M^\odot$ as follows. As long as $\lef(x) \in \lastextent(P)$ and $x$ is on $s$ we add $x$ to $M^\odot$ and set $x=\nex^\odot(x,\alpha)$.
\end{enumerate}

\paragraph{Computing $M^\ast$}
  We find for each segment $s = (v, w)$ all $\ast$-transition nodes for $\delta(P, \alpha)$. 
  \begin{enumerate}
   
    \item  If $s$ is not a $\ast$-segment we stop. 
     
    \item\label{alg:firstast}  \emph{Find first $\ast$-transition node on $s$:} Compute the predecessor $q^-$ and successor $q^+$ of $\bots(s)$ in $\Pos_\alpha$ and let $v$ be the lowest of $\lca(\bots(s), q^-)$ and $\lca(\bots(s), q^+)$.  If $v$ is not on $s$ stop.  If $\alpha \not\in \lab(v)$ and $\parentstar(v)\in \lastextent(P)$ we add $v$ to $M^\ast$.

   \item\label{alg:nextast}  \emph{Repeatedly find next  transition node by following $\nex^{\ast}$-pointers: } We first compute $x = \firstlabel(v, \alpha)$. We find all $\ast$-transition nodes by repeatedly following $\nex^{\ast}$-pointers   from $x$  and adding the visited nodes to $M^\ast$ as follows. As long as $\parentstar(x) \in \lastextent(P)$ and $x$ is on $s$ we add $x$ to $M^\ast$ and set $x=\nex^\ast(v,\alpha)$.
   \end{enumerate}

\paragraph{Complexity}
We first analyze the time used to find $M^\odot$. We use $O(|M^\odot|)$ time to follow pointers. Additionally, we use $O(|T|\log \log m) = O(|P|\log \log m)$ time to compute the first label queries, as we do one first label query on each of the $|T|$ segments. The time to check if a node is $\odot$-live is constant. For all the bottom nodes of the segments the information is stored in the tree and for all the other nodes $v$ we can check in constant time if $\lef(v)\in \lastextent(P)$ as described in the end of Section~\ref{sec:transitiontree-construct} as $\lef(v)$ will always be on the current segment. Thus the total time used is $O(|M^\odot| + |P|\log \log m)$. 
Similarly, we use time $O(|M^\ast| + |P|\log \log m)$ to compute $M^\ast$. Next we analyze the size of $M^\odot$ and $M^\ast$.

\begin{lemma}\label{lem:size-modot}
$|M^\odot| = O(|\delta^\odot(P,\alpha)| + |P|)$. 
\end{lemma}
\begin{proof} 
We will prove that at most one node from $M^\odot$ from each segment is not in $N^\odot(P,\alpha)$. Recall that  $N^\odot(P,\alpha)$ consists of all the $\odot$-nodes $v$ that have $\lef(v)\in \lastextent(P)$ and $\rig(v) \in \firstextent(\Pos_\alpha)$.

Any node $u\in s$ added to $M^\odot$ in step~\ref{alg:nextodot} except the first one has $\delta(u,\alpha) \neq \emptyset$, since they were found using $\nex^{\odot}$-pointers. 
Thus, $\rig(u)\in \firstextent(\Pos_\alpha)$.  A node is only added if it is $\odot$-live, i.e., $\lef(u)\in \lastextent(P)$. Thus $u\in N^\odot(P,\alpha)$.
Therefore, only the first node found on each segment---the node from step~\ref{alg:firstodottrans}---might not be in $N^\odot(P,\alpha)$. Since there are $O(|P|)$ segments in $T$ we have $|M^\odot| = O(|N^\odot(P,\alpha)| + |P|)$. By Lemma~\ref{lem:sizeNodotNast} we have $|N^\odot(P,\alpha)| \leq |\delta^\odot(P,\alpha)| + |P|$ and thus $|M^\odot| = O(|\delta^\odot(P,\alpha)| + |P|)$.
\end{proof}
The argument for the size of $M^\ast$ is similar, but here we show that the number of nodes in $M^\ast$ that are not in $N^\ast(P,\alpha)$ is at most $2|P|$.
\begin{lemma}\label{lem:size-mast}
$|M^\ast| =O(|\delta^\ast(P,\alpha)| + |P|)$.
\end{lemma}
\begin{proof} We will prove that at most two nodes from $M^\ast$ from each segment is not in $N^\ast(P,\alpha)$. Recall that $N^\ast(P,\alpha)$ consists of all the nodes $v$ such that there exists a $q \in \Pos_\alpha$  and a $p \in P$  such that  $v= \lca(p,q)$  and  $\parentstar(v) \in \lastextent(p) \cap \firstextent(q)$. 

Any node $u$ added to $M^\ast$ by following $\nex^\ast$-pointers has $\delta^\ast(u,\alpha) \neq \emptyset$ and $\parentstar(u) \in \lastextent(P)$. Furthermore, if $\lef(u)$ is on the segment $s$, then there exists a position $q\in \rig(u) \cap \delta^\ast(u,\alpha)$ and a position $p \in P \cap \lef(u)$.  This implies that $u$ is in $N^\ast(P,\alpha)$. The argument for the case where $\rig(u)$ is on the segment $s$ is symmetric. 

At most two other nodes are added to $M^\ast$ for each segment (the first two nodes are added to $M^\ast$ on each segment). It follows that the total number of nodes in  $M^\ast$ is at most $ |\delta^\ast(P,\alpha)| + 2|P|$.
By Lemma~\ref{lem:sizeNodotNast} we have $|N^\ast(P,\alpha)| \leq |\delta^\ast(P,\alpha)| + |P|$ and thus $|M^\ast| = O(|\delta^\ast(P,\alpha)| + |P|)$. 
\end{proof}

Combining Lemmas~\ref{lem:size-modot} and~\ref{lem:size-mast} and the above discussion, the total time to compute $M^\odot$ and $M^\ast$ is $O(|\delta^\odot(P,\alpha)| + |P|\log\log m)$ and $O(|\delta^\ast(P,\alpha)| + |P|\log\log m)$, respectively.

\paragraph{Correctness}
We argue that the sets $M^{\odot}$ and $M^{\ast} $ include all transition nodes for $\delta(P, \alpha)$.  
We need the following lemma, which follows from the path structure of the last extent sets. 
\begin{lemma}\label{lem:odotsegmentlive}
Let $u$ be an $\odot$-live node and let $s$ be the segment in $T$ containing $u$.  
All $\odot$-nodes below $u$ on $s$ with $\lef(u)$ on $s$ are also $\odot$-live.
\end{lemma}
\begin{proof}
If $\lef(u)$ is not on $s$ then $u$ is a branching node in $T$ by construction of $T$. Thus $u = \bots(s)$ and it is trivially true since $u$ has no descendants on $s$. We will prove the case $\lef(u)$ on $s$ by contradiction.  Since $u$ is $\odot$-live there  exists a node $p \in \lef(u) \cap P$ such that $u \in \lastextent(p)$. Furthermore, $\bots(s) $ is an ancestor of some node $p'$ in $P$. Let $w = \lca (\bots(s), p)$. If $p'=p$ then $w = \bots(s)$. If $p'\neq p$ then $w$ is a branching node in $T$. Since there are no branching nodes internally on a segment it follows that $w=\bots(s)$. By Lemma~\ref{lem:pathextent} all nodes on $s$ are in $\lastextent(p)$.
\end{proof}

We are now ready to prove that $N^\odot(P,\alpha)$ is contained in $M^\odot$. Let $u$ be a node in $N^\odot(P,\alpha)$. There are two main cases in the proof depending on whether $u$ labeled $\alpha$ or not. If not, then we show that $u = \lca(x,q)$ in step~\ref{alg:firstodottrans}. If $u$ is labeled $\alpha$, then we show, that either $u$ is the first $\odot$-node on the segment, in which case it is added as the first node in step~\ref{alg:nextodot}, or it is a node on the path induced by the $\nex^\odot$-pointers. As we proved in Lemma~\ref{lem:odotsegmentlive} above, all the nodes on this path below $u$ are also $\odot$-live and we show by induction that all the $\odot$-live nodes from $s$ on this path are added to $M^\odot$ in step~\ref{alg:nextodot}.

\begin{figure}[t]
\centering  
\includegraphics[scale = 0.4]{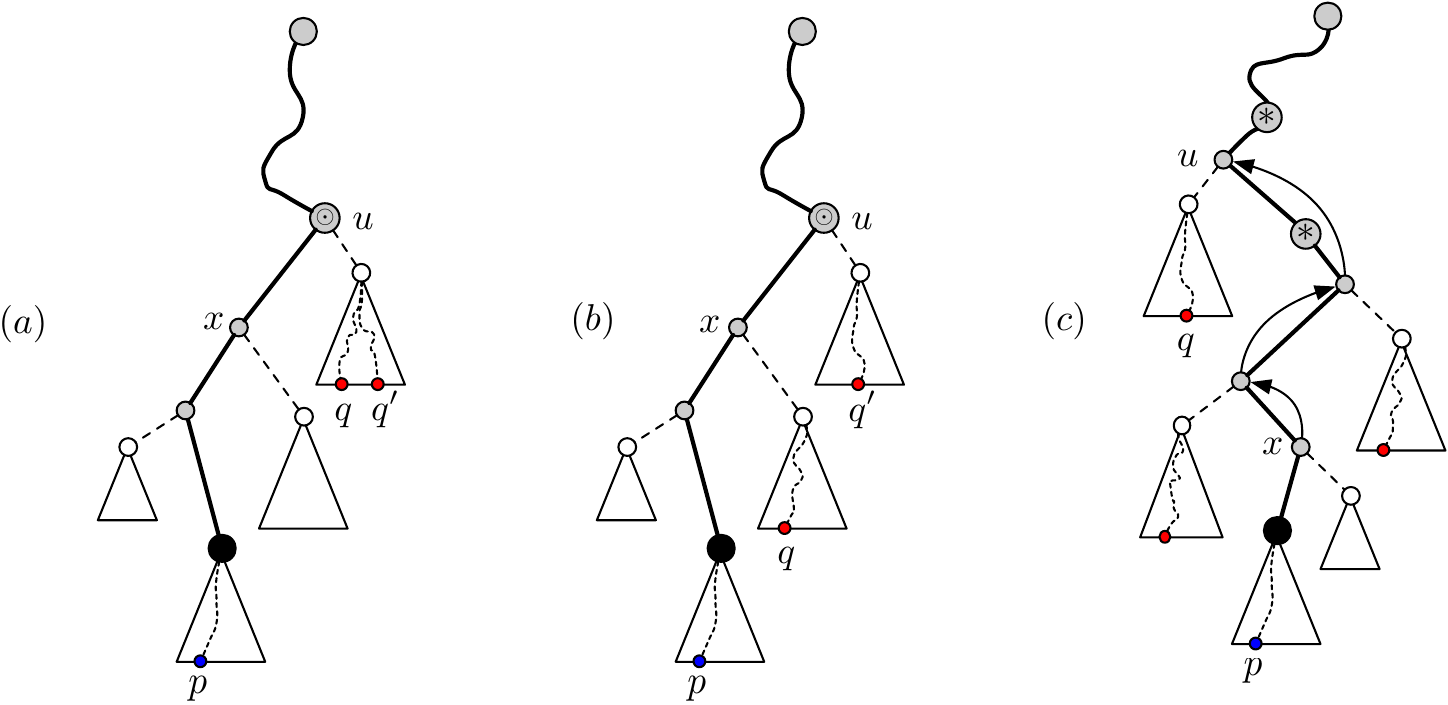}
  \caption{The thick edges are the edges on the segment $s$ and the black node is $\bots(s)$. In (a) both $q$ and $q'$ are positions in the right subtree of $u$, and $u = \lca(x,q)$. In (b)  the positions $q$ and $q'$ are in different subtrees and node $u$ has label $\alpha$.  In (c) the arrows indicates the $\nex^\ast$-pointers. }
  \label{fig:nodot-in-modot}	
\end{figure}

\begin{lemma}
$N^\odot(P,\alpha) \subseteq M^{\odot}$.
\end{lemma}
\begin{proof} 
Let $u \in N^\odot(P,\alpha)$. By definition $u$ is $\odot$-live,  $\rig(v) \in \firstextent(\Pos_\alpha)$, and $\delta(u,\alpha) \neq \emptyset$.
Let $s$ be the segment in $T$ containing $u$. 
We want to show that $u$ is added to $M^\odot$.  Define $x = \parentodot(\bots(s))$ and $q$  as the successor of $x$ in $\Pos_\alpha$ as in step~\ref{alg:firstodot} and~\ref{alg:firstodottrans} of the algorithm.  There are two cases depending on  whether $u$ has label $\alpha$ or not. 

\proofsubparagraph{Case 1: $\alpha\not\in \lab(u)$}
We first show that $u = \lca(x,q)$.  Since $u$ is an $\odot$-node we have that $u$ is an ancestor of $x$. If $u$ is a proper ancestor of $x$ then $x \in R(\lef(u))$. 
Since $\rig(u) \in \firstextent(\Pos_\alpha)$ there is a position $q'\in \Pos_\alpha$ such that $q' \in \Pos(\rig(u))$. This implies that $q \leq q'$. 
If $q \in \Pos(\rig(u))$ then $u =  \lca(x,q)$ and we are done (see Figure~\ref{fig:nodot-in-modot}(a)). Assume $q \not\in \Pos(\rig(u))$ and let $w = \lca(x,q) \neq u$.   
 Since $q < q'$ we have $u \prec w$. But then $q \in \Pos(w) \subset \Pos(\lef(u))$. Thus $u= \lca(q,q')$ and $u$ has label $\alpha$ which is a contradiction (see Figure~\ref{fig:nodot-in-modot}(b)). 
 It follows  that $u = \lca(x,q)$. Thus, since $u$ is $\odot$-live  and not labeled $\alpha$, $u$ is added to $M^\odot$ in step~\ref{alg:firstodottrans}.

\proofsubparagraph{Case 2: $\alpha\in \lab(u)$}
There are two subcases. 
In the first case $u = \parentodot(\bots(s))$. Then $u= x$ in step~\ref{alg:firstodot}. Since $\rig(u) \in \firstextent(\Pos_\alpha)$ we have $u =v =  \lca(x,q)$ in step~\ref{alg:firstodottrans}. Since $u$ is labeled $\alpha$ we have $u = \firstlabel(v,\alpha)= \firstlabel(u,\alpha)$. Since $\lef(u)\in \lastextent(P)$  we add $u$ to $M^\odot$ in the first iteration in step~\ref{alg:nextodot}.

 In the other case $u\neq \parentodot(\bots(s))$. 
 Let 
 \[S_\alpha = \{w  \in s \mid w \textrm{ is a } \odot\textrm{-live node}, \lef(w) \in s,  \alpha \in \lab(w) \textrm{ and }  \delta^\odot(w, \alpha) \neq \emptyset  \}\;.\] 
 Clearly, $u \in S_\alpha$, and thus $u \in M^\odot$ follows from the following claim.   
 
 \begin{claim}
 Let  $w$ be a node in $S_\alpha$. Then $w$ is added to $ M^\odot$ in step~\ref{alg:nextodot}.
 \end{claim}
\begin{claimproof} 
We prove the claim using induction on the height of $w$. For the base case let  
$w$ be the lowest node in $S_\alpha$.  Since $w$ is $\odot$-live it follows from Lemma~\ref{lem:odotsegmentlive} that $\parentodot(\bots(s))$ is also $\odot$-live and thus we continue from step~\ref{alg:firstodot}  to step~\ref{alg:firstodottrans}. Since $\rig(w) \in \firstextent(\Pos_\alpha)$  we have $w \preceq  v$, where $v= \lca(x,q)$ as computed in step~\ref{alg:firstodottrans}. Now either $w = \firstlabel(v, \alpha)$ and then $w$ is added to $M^\odot$ in the first iteration in step~\ref{alg:nextodot}. Otherwise, $w\prec \firstlabel(v, \alpha)$. Since $w$ is the lowest node in $S_\alpha$ then  $x = \firstlabel(v, \alpha)$ has its $\nex^\odot$-pointer pointing to $w$. Now either $x$ is an $\odot$-live node, in which case $x$ is set to $w$ in the end of the first iteration of step~\ref{alg:nextodot}. Otherwise,  $x$ is set to $w$ in the end of step~\ref{alg:firstodottrans}. In either case,  $w$ is added to $M^\odot$ in step~\ref{alg:nextodot}. 

Induction step: Let $w\in S_\alpha$ be a node that is not the lowest node in $S_\alpha$. Let $w'\in S_\alpha$ be the highest node in $S_\alpha$ that is a proper descendant of $w$.  Then $w'$ points to $w$. Let $v = \nex^\odot(w',\alpha)$. By definition of the $\nex^\odot$-pointers $v \in S_\alpha$ and $v \prec w'$. Now, since $w' \in R(\lef(w))$, $\alpha\in \lab(w)$ and $\delta^\odot(w,\alpha) \neq \emptyset$, we have $v =\nex^\odot(w',\alpha)\succeq w$.  If $v \neq w$ then $w' \succ v \succ w$ contradicting that $w'$ is the highest descendant of $w$ in $S_\alpha$. Thus $v = w$.
By the induction hypothesis, $w'$ was added to $ M^\odot$ in step~\ref{alg:nextodot}, whereafter we follow the $\nex^\odot$-pointer to $w$ and add $w$ in the next iteration. 
\end{claimproof}
\end{proof}

We will now prove that $N^{\ast}(P,\alpha)\subseteq M^{\ast} $. Here is an outline of the proof.  Recall that for any node $u$ in $N^\ast(P,\alpha)$, there exists a position $q$  in $\Pos_\alpha$ and a position $p$ in $P$, such that $u$ is the $\lca$ of $q$ and $p$ and $\parentstar(u)$ is in both  $\firstextent(q)$ and  $\lastextent(p)$.  We first prove that the segment $s$ containing $u$ is a $\ast$-segment. Then it follows easily that if $u=\bots(s)$ then $u$ is added to $M^\ast$ in step~\ref{alg:firstast} or~\ref{alg:nextast}. If $u$ is not the bottom node on $s$,  then due to the properties of the transition tree the child of $u$ not on $s$ contains $q$ in its subtree. We then show that if $u$ is not labeled $\alpha$ then it is added to $M^\ast$ in step~\ref{alg:firstast}. Otherwise, it is either added as the first node in step~\ref{alg:nextast}, or it is a node on the path induced by the $\nex^\ast$-pointers on $s$. 
By similar arguments as those in the previous proof all nodes below $u$ on this path has their $\parentstar$-node in $\lastextent(p)$, the first node on the path is $x$ found by a firstlabel query in step~\ref{alg:nextast}, and thus $u$ is added to $M^\ast$ in step~\ref{alg:nextast}. See Figure~\ref{fig:nodot-in-modot}.

\begin{lemma}
$N^{\ast}(P,\alpha)\subseteq M^{\ast} $.
\end{lemma}
\begin{proof}
Let $u$ be a node in $N^\ast(P,\alpha)$. Then there exists a position $p\in P$ and a position $q \in \Pos_\alpha$ such that $u = \lca(p,q)$ and $\parentstar(u)\in \lastextent(p)\cap \firstextent(q)$.
Let $s$ be the segment $u$ is on. We first prove that $s$ is a $\ast$-segment:  If $u = \bots(s)$ then  $\parentstar(\bots(s)) = \parentstar(u) \in \lastextent(P)$ and thus $s$ is a $\ast$-segment. Otherwise, $u \prec \bots(s)$. By construction of $T$ we have $p \in \Pos(\bots(s))$ and thus by Lemma~\ref{lem:pathextent} $\bots(s) \in \lastextent(p)$, since $\parentstar(u)$ is an ancestor of $\bots(s)$. 

We now prove that $u \in M^\ast$. Let $v$ be the lowest of $\lca(\bots(s),q-)$ and $\lca(\bots(s),q+)$ as computed in step~\ref{alg:firstast}. If $u = \bots(s)$ then $q\in \Pos(\bots(s))$ and thus $v = \bots(s) = u$. If $\alpha \not \in \lab(u)$ then $u$ is added to $M^\ast$ in step~\ref{alg:firstast}. Otherwise, it is added in step~\ref{alg:nextast}.

If $u\neq \bots(s)$,  let $u_s$ be the child of $u$ that is on $s$ and let $u_o$ be the child not on $s$. 
By construction of $T$ we have $p \in \Pos(u_s)$ and $p \in \Pos(\bots(s))$, since all bottom nodes of segments in $T$ have a position from $P$ in their subtree and any node that is the $\lca$ of two nodes in $P$ is a branching node in $T$. Thus $q \in \Pos(u_o)$.
There are two cases depending on whether $u$ is labeled $\alpha$ or not.  
 
 \proofsubparagraph{Case 1: $\alpha\not\in \lab(u)$} 
 Since $q \in \Pos(u_o)$ we have $\Pos(u_s) \cap \Pos_\alpha = \emptyset$.  Thus, either $q+$ or $q-$ is in $\Pos(u_o)$ and then $v= u$. Therefore, $u$ is added to $M^\ast$ in step~\ref{alg:firstast}.

\proofsubparagraph{Case 2: $\alpha\in \lab(u)$} 
Then $u$ is the $\lca$ of two positions in $\Pos_\alpha$, which implies that  there exists a position labeled $\alpha$ in $\Pos(u_s)$. It follows that either $q-$ or $q+$ is in $\Pos(u_s)$ and thus $u \preceq v$. If $u=v$  then $u$ is added to $M^\ast$ in the first iteration in step~\ref{alg:nextast}. Otherwise, $u \prec v$. Let $x = \firstlabel(v,\alpha)$. If $u=x$  then $u$ is added to $M^\ast$ in the first iteration in step~\ref{alg:nextast}. 

If $u \neq x$ then $u \prec x$. Since $q \in \Pos(u_o)$ and $\parentstar(u) \in \firstextent(q)$ we have $q \in \delta^\ast(u,\alpha)$. Therefore, the highest proper descendant of $u$ labeled $\alpha$ has its $\nex^{\ast}$-pointer pointing to $u$.  It follows that there is a chain of  $\nex^{\ast}$-pointers from $x$ to $u$.
It remains to show that for all the nodes in this chain $x = x_0, \ldots, x_k = u$ we have $\parentstar(x_i) \in \lastextent(P)$. If this is true we will reach $u$ in step~\ref{alg:nextast}. Since $p \in \Pos(\bots(s))$ then all nodes on $s$ are ancestors of $p$. Since $u \preceq x_i$ we have $\parentstar(u) \preceq \parentstar(x_i)$. By Lemma~\ref{lem:pathextent}  we have $ \parentstar(x_i) \in \lastextent(p)$.
\end{proof}

\subsubsection{Computing the Intervals}
We now compute the lists $L^\odot$ and $L^\ast$ of intervals for the nodes in $M^\odot$ and $M^\ast$, respectively. We do this by processing the nodes $M^\odot$ and $M^\ast$ in inorder using a depth-first left-to-right inorder traversal of $T$.

First, we compute for each node $v$ in $M^\odot$ the range $[l_v, r_v]$ of positions labeled $\alpha$ that are descendants of $\rig(v)$. If $v$ is labeled $\alpha$ the range $[l_v, r_v]$ is stored at $v$ and otherwise we use the predecessor data structure to compute it using the range stored at $\rig(v)$. Similarly, we compute for each node in $M^\ast$ the range $[l_v,r_v]$ of positions labeled $\alpha$ that are descendants of~$v$. 

\paragraph{Computing $L^\odot$} We compute the list of intervals $L^\odot$ by a depth-first left-to-right inorder traversal of $T$. We maintain a stack $S$ keeping track of the deepest node not finished and a  counter $\ell$ equal to the left starting point of the currently open interval. If there is no open interval $\ell = -1$. 
Initially, $S=\emptyset$ and $\ell = -1$.

 For each node $v\in M^\odot$ in inorder do the following:
\begin{itemize}
\item When we meet $v$ in the traversal after traversing the left subtree of $v$: If $\ell \neq -1$ append $(\ell, l_v-1, \tops(S))$ to~$L^\odot$. Set $\ell = l_v$ and push $v$ onto the stack $S$. 
\item When we finish the traversal of the subtree containing $v$: Note that in this case $\tops(S) = v$.  If $\ell \leq r_v$ append $(\ell, r_v, v)$ to $L^\odot$. Pop $v$ from $S$. If the stack is now empty set $\ell = -1$, otherwise set $\ell = r_v+1$.
\end{itemize}
Note that $\rig(v)$ might not be in $T$, in which case the two steps for $v$ follow immediately after each other.

\paragraph{Computing $L^\ast$} We maintain a stack $S$ and counter $\ell$ as before. 

For each node $v\in M^\ast$ in inorder do the following:
\begin{itemize}
\item First time we meet $v$ in the traversal: If $\ell \neq -1$ append $(\ell, l_v-1, \tops(S))$ to $L^\ast$. Set $\ell = l_v$ and push $v$ onto the stack $S$. 
\item Last time we meet $v$ in the traversal: Note that in this case $\tops(S) = v$.  If $\ell \leq r_v$ append $(\ell, r_v, v)$ to $L^\ast$. Pop $v$ from $S$. If the stack is now empty set $\ell = -1$, otherwise set $\ell = r_v+1$.
\end{itemize}

\paragraph{Complexity}
To compute the ranges use time $O(|M^\odot| + |P|\log \log m )$ as at most one node in $M^\odot$ on each segment is not labeled $\alpha$. We will store the nodes in $M^\odot$ in increasing order of depth for each segment. This is easy to maintain as we find them in order of decreasing depth. This way we can do the depth-first left-to-right traversal on the nodes of $M^\odot$ in $T$ in linear time in the size of $M^\odot$. Thus computing $L^\odot$ takes time $O(|M^\odot|)$. 

Similarly, we use time $O(|M^\ast|)$ to compute $L^\ast$. In summary, we have the following lemma. 

\begin{lemma}\label{lem:computetransitionnodes}
    The sets $M^\odot$ and $M^\ast$, and the lists $L^\odot$ and $L^\ast$, can be computed in time $O(|\delta(P,\alpha)| + |P|\log\log m)$.
\end{lemma}

\subsection{Internal Transitions}\label{sec:computing-internal}

We will compute and return the state-set transition by computing the internal transitions on the nodes in $L^\odot$ and $L^\ast$. 
Next, we show how to compute internal transitions efficiently using 3-sided range queries. We assume that the range $[l,r]$ is given as indexes in $A_\alpha$.

\paragraph{Computing $\delta^\odot_{[l,r]}(v, \alpha)$.} Given an $\odot$-node $v$, a character $\alpha$, and a range $[l,r]$ we compute $\delta^\odot_{[l,r]}(v, \alpha)$ as follows. We perform a 3-sided range reporting query $(l,r,\depth(\rig(v))$ on $D_\alpha$. That is, we return all positions in $D_\alpha[l,r]$ with a value less than or equal to $\depth(\rig(v))$. This can be done by a standard technique of recursively applying range minimum queries as follows. Let $j$ be the position returned by  range minimum query on $D_\alpha[l,r]$. If $D_\alpha[j]\leq \depth(\rig(v))$ return $A_\alpha[j]$ and recurse on the ranges $[l,j-1]$ and $[j+1,r]$. We stop if this is not the case or if the range is empty.

For instance, suppose we compute $\delta^\odot_{[2,4]}(v_2, a)$ in our example in Figure~\ref{fig:exampleexpression}. The range $[2,4]$ in $A_\alpha$ corresponds to the positions $p_2$, $p_3$, and $p_5$.  
We find the highest first extent in to be $f = v_4$ corresponding to both $p_2$ and $p_3$. Suppose $j = 3$ corresponding to $p_3$. Then we compare $f$ with $\rig(v_2) = v_4$ and since $f$ is an ancestor of $v_4$ we report $p_3$ and repeat on the subarrays $[2,2]$ and $[4,4]$. On $[2,2]$ we return $p_2$, while on $[4,4]$ we do not get a position since $v_8$ is a proper descendant of $v_4$. 

Note, that we can get the output in sorted order if we first recurse on the range $[l,j-1]$, then report $A_\alpha[j]$, and then recurse on the range $[j+1,r]$.

The algorithm returns all positions $q$ with label $\alpha$ in $[l,r]$ such that $\rig(v) \in \firstextent(q)$ and is thus correct. 
Each recursive call uses constant time and we repeat at most $2|\delta^\odot_{[l,r]}(v, \alpha)| + 1$ times. Hence, in total we use $O(1+ |\delta^{\odot}_{[l,r]}(v, \alpha)|)$ time.

\paragraph{Computing $\delta^\ast_{[l,r]}(v,\alpha)$.} To compute an internal transition for the $\ast$-case we do a 3-sided range reporting query $(l,r,\depth(\parentstar(v))$. Correctness follows as above and the time and space bounds are the same.

In summary, we have the following result. 

\begin{lemma}\label{lem:internaltransitioninterval}
	Let $R$ be a regular expression of size $m$. Given $R$ we can build a data structure in $O(m)$ space and preprocessing time, such that given a node $v$ in $R$, a character $\alpha \in \Sigma$, and a range $[l,r]$ we can compute 
 $\delta^\odot_{[l,r]}(v,\alpha)$ in sorted order in time $O(1 + |\delta^\odot_{[l,r]}(v,\alpha)|)$ and  $\delta^{\ast}_{[l,r]}(v, \alpha)$ in sorted order in time $O(1+ |\delta^{\ast}_{[l,r]}(v, \alpha)|)$.
\end{lemma}

\subsection{Computing State-Set Transitions}\label{sec:mainalg}
Given a set of positions $P$ and a character $\alpha$, we compute the state-set transition $\delta(P,\alpha)$ as follows. For simplicity, we assume that the positions in $P$ are sorted according to their left-to-right order since otherwise we can sort them in additional $O(|P|\log \log m)$ time using integer sorting.  
The final algorithm for computing fast state-set transitions is as follows. 

\begin{enumerate}
    \item\label{alg:transitiontree} First, we construct the transition tree and all the information from Lemma~\ref{lem:transitiontree-compute} as in Subsection~\ref{sec:transitiontree-construct}.
    \item\label{alg:transitionnodes} We then compute $M^\odot$, $M^\ast$, $L^\odot$, and $L^\ast$  as in Subsection~\ref{sec:computetransitionnodes}.
    \item~\label{alg:transitions} Finally, we compute and return the state-set transition by computing
	\[D^\odot =  \bigcup_{(l,r,u) \in L^{\odot}} \delta^\odot_{[l,r]}(u, \alpha) 
	\qquad \textrm{ and } \qquad
	D^\ast = \bigcup_{(l,r,u) \in L^{\ast}} \delta^\ast_{[l,r]}(u, \alpha).\]
 by processing $L^\odot$ and $L^\ast$ from left-to-right using the procedure from Subsection~\ref{sec:computing-internal} that computes  $\delta^\odot_{[l,r]}(u, \alpha)$ and $\delta^\ast_{[l,r]}(u, \alpha)$ in sorted order.  Since both lists $L^\odot$ and $L^\ast$ are sorted the lists $D^\odot$ and $D^\ast$ are also sorted. We merge these two lists to get the final output. 
\end{enumerate}

\paragraph{Analysis of the algorithm}
Step~\ref{alg:transitiontree} uses $O(|P|)$ time by Lemma~\ref{lem:transitiontree-compute}, and 
Step~\ref{alg:transitionnodes} uses $O(|\delta(P,\alpha)| + |P|\log\log m)$ time by Lemma~\ref{lem:computetransitionnodes}.
By Lemma~\ref{lem:internaltransitioninterval} the time to compute all internal transitions in step~\ref{alg:transitions} is 
\begin{equation}\label{eq:time}
O\left(\sum_{(l,r,u) \in L^{\odot}} (1 + |\delta^\odot_{[l,r]}(u, \alpha) |) + \sum_{(l,r,u) \in L^{\ast}} (1 + |\delta^\ast_{[l,r]}(u, \alpha) |) \right)
\end{equation}
The length of list $L^\odot$ is $O(|M^\odot|)$ as each interval endpoint is due to meeting a node in $M^\odot$ for the first or last time in the traversal. Similarly, $|L^\ast| = O(|M^\ast|)$. Hence, \eqref{eq:time} is $O(|M^\odot| + |M^\ast| + |\delta(P,\alpha)|)$. 
The time to merge the two lists is linear in the total size of the lists since the lists are sorted. Thus the  time for step~\ref{alg:transitions} is $O(|M^\odot| + |M^\ast| + |\delta(P,\alpha)|)$. Plugging in the bounds from Lemma~\ref{lem:size-modot} and \ref{lem:size-mast} 
we get a total running time of $O(|\delta(P,\alpha)| + |P|\log\log m)$.


\paragraph{Correctness}
We have already proved that $N^\odot(P,\alpha) \subseteq M^\odot$ and $N^\ast(P,\alpha) \subseteq M^\ast$. 
Thus, by Lemma~\ref{lem:transitiondecomp} the set of positions computed in step 4 is $\delta(P, \alpha)$. There are at most $|P|$ nodes in  $M^\odot \setminus N^\odot(P,\alpha)$. These nodes are all in $\lastextent(P)$ and thus the set of internal transitions on these will all be in $\delta(P,\alpha)$. Similarly, for the nodes in $M^\ast\setminus  N^\ast(P,\alpha)$.

It remains to show that the partition into lists is correct. We show that a position $q \in \bigcup_{v \in M^\odot} \delta^\odot(v,\alpha)$ is covered by the lowest node in $M^\odot$ such that $q \in \delta^\odot(u,\alpha)$, i.e., where covered means that $q \in \delta^\odot_{[l,r]}(u,\alpha)$ for some $l,r$ such that $(l, r, u) \in L^\odot$. This follows easily from the properties of the inorder traversal. The arguments for $q \in \cup_{v \in M^\ast} \delta^\ast(v,\alpha)$
are similar. For completeness, the full proof is shown below. 

\begin{lemma}
We have 
\[
\bigcup_{v \in M^\odot} \delta^\odot(v,\alpha)\quad = \bigcup_{(l,r,u) \in L^\odot} \delta_{[l,r]}^\odot(u,\alpha) 
\qquad \textrm{ and } \qquad
\bigcup_{v \in M^\ast} \delta^\ast(v,\alpha) \quad = \bigcup_{(l,r,u) \in L^\ast} \delta_{[l,r]}^\ast(u,\alpha) \;.
\]
\end{lemma}
\begin{proof} We split the proof into two cases. 

\proofsubparagraph{Case 1: $ \bigcup_{v \in M^\odot} \delta^\odot(v,\alpha)=\bigcup_{(l,r,u) \in L^\odot} \delta_{[l,r]}^\odot(u,\alpha)$}  
We are only adding subranges of the range of $\rig(v)$ for any node in $v\in M^\odot$ to $L^\odot$. This immediately implies that $\cup_{(l,r,u) \in L^\odot} \delta_{[l,r]}^\odot(u,\alpha) \subseteq \cup_{v \in M^\odot} \delta^\odot(v,\alpha)$.

For the other direction, let
$q$ be a position in  $ \cup_{v \in M^\odot} \delta^\odot(v,\alpha)$ and let  $u$ be the deepest node in $M^\odot$ such that $q \in \delta^\odot(u,\alpha)$.
We will show that $q \in \delta^\odot_{[l,r]}(u,\alpha)$ for some $l,r$ such that $(l, r, u) \in L^\odot$.

Let $\predu(u) = \arg\max_{x \in M^\odot}\{r_x < q\}$ and $\succu(u) =\arg\min_{x \in M^\odot}\{l_x > q\}$. If neither $\predu(u)$ nor $\succu(u)$ are in $T(\rig(u))$ then nothing happens with $\ell$ and the stack $S$ after the step where we add $u$ to the top of the stack until we leave $u$ the last time. At this point $\tops(S) = v$ and $\ell = l_u$ and $(\ell, r_u, u) = (l_u, r_u, u)$ is appended to $L^\odot$.
If  $\predu(u)\in T(\rig(u))$ then when we leave $\predu(u)$,  node $u$ will be on top of the stack and $\ell = r_{\predu(u)}+1$. If also $\succu(u) \in T(\rig(u))$ then $\succu(u)$ is the next node from $M^\odot$ we process in our inorder traversal. When we meet $\succu(u)$ before traversing its right subtree we add $(\ell, l_{\succu(u)}-1, \tops(S)) = (r_{\predu(u)}+1, l_{\succu(u)}-1,u)$ to $L^\odot$. Since $r_{\predu(u)}+1\leq q \leq l_{\succu(u)}-1$  we have $q \in \delta^\odot_{[r_{\predu(u)}+1,l_{\succu(u)}+1]}(u,\alpha)$. If $\succu(u)\not \in T(\rig(u))$ then the next change we perform is when leaving $u$. Here we add $(\ell, r_{u}-1, u)$ to $L^\odot$. 
If only $\succu(u) \in T(\rig(u))$  then we have $\tops(S) = u$ and $\ell = l_u$ when we process $\succu(u)$ the first time (after visiting its left subtree). Then we add  $(l_u, l_{\succu(u)}-1, u)$ to $L^\odot$.

\proofsubparagraph{Case 2: $\bigcup_{v \in M^\ast} \delta^\ast(v,\alpha) = \bigcup_{(l,r,u) \in L^\ast} \delta_{[l,r]}^\ast(u,\alpha)$} 
The arguments are similar to case 1. We are only adding subranges of the range of $v$ for any node in $v\in M^\ast$ to $L^\ast$. This immediately implies that $\cup_{(l,r,u) \in L^\ast} \delta_{[l,r]}^\ast(u,\alpha) \subseteq \cup_{v \in M^\ast} \delta^\ast(v,\alpha)$.

For the other direction let $q$ be a position in  $ \bigcup_{v \in M^\ast} \delta^\ast(v,\alpha)$ and let  $u$ be the deepest node in $M^\ast$ such that $q \in \delta^\ast(u,\alpha)$.
We will show that $q \in \delta^\ast_{[l,r]}(u,\alpha)$ for some $l,r$ such that $(l, r, u) \in L^\ast$.

If there are no other nodes from $M^\ast$ than $u$ in $T(u)$ then nothing happens between the first and last time we meet $u$ in the traversal and $(l_u, r_u, u)$ is appended to $L^\odot$ when we meet $u$ the last time. Similarly to case 1, let $u_{\predu(u)} = \arg\max_{x \in M^\ast}\{r_x < q\}$ and $\succu(u) =\arg\min_{x \in M^\ast}\{l_x > q\}$. If $\predu(u) \in T(u)$ then when we leave $\predu(u)$,  node $u$ will be on top of the stack and $\ell = r_{\predu(u)}+1$. If also $\succu(u) \in T(u)$ then $\succu(u)$ is the next node from $M^\ast$ we process in our traversal. When we meet $\succu(u)$ the first time we add $(\ell, l_{\succu(u)}-1, \tops(S)) = (r_{\predu(u)}+1, l_{\succu(u)}-1,u)$ to $L^\ast$. Since $r_{\predu(u)}+1\leq q \leq l_{\succu(u)}-1$  we have $q \in \delta^\ast_{[r_{\predu(u)}+1,l_{\succu(u)}+1]}(u,\alpha)$. If $\succu(u)\not \in T(u)$ then the next change we perform is when leaving $u$. Here we add $(\ell, r_{u}-1, u)$ to $L^\ast$. 
If only $\succu(u) \in T(u)$  then we have $\tops(S) = u$ and $\ell = l_u$ when we process $\succu(u)$ the first time (after visiting its left subtree). Then we add  $(l_u, l_{\succu(u)}-1, u)$ to $L^\ast$.
\end{proof}

\noindent In summary, we have the following result.

\begin{lemma}
	Given a regular expression $R$ of size $m$, we can build a data structure in $O(m)$ space and preprocessing time such that given any set of positions $P$ in $R$ and character $\alpha \in \Sigma$, we can compute $\delta(P, \alpha)$ in $O(|P|\log \log m + |\delta(P, \alpha)|)$ time.
\end{lemma}

\section{Speeding Up State-Set Transitions}\label{sec:speedup}
We now show how to improve the run time of computing a state-set transition $\delta(P, \alpha)$ from $O(|P|\log \log m + |\delta(P, \alpha)|)$ to  $O(|P|\log \log \frac{m}{|P|} + |\delta(P, \alpha)|)$ while still using linear space. Finally, we show how to use this to obtain the main results of Theorems~\ref{thm:regexmatching}. We now require that the input positions in $P$ are sorted and the output positions in $\delta(P, \alpha)$ are reported in sorted order.

First, observe that the $\log \log m$ factor is from computing $|P|$ predecessor queries and $|P|$ first label queries in steps 2 and 3 in the main algorithm in Section~\ref{sec:mainalg}. The first label queries in turn are  reduced to computing $O(|P|)$ predecessor queries on the Euler tour of $R$~\cite{Die89}. In both cases, we need to answer a batch of $b = \Theta(|P|)$ predecessor queries on a set of size $t = \Theta(|\Pos_\alpha|)$ from a universe of size $u = \Theta(m)$. The batch is provided in sorted order and the output should also be sorted. 

We use a simple two-level data structure as follows.  We first partition the universe into $t$ intervals of size $u/t$ (except possibly the last which may be smaller). For each interval, we store a predecessor data structure over the subset of the elements in the interval using a reduced universe of size $u/t$. Furthermore, we also store a pointer to the nearest non-empty smaller interval. Using the same predecessor data structure as in Section~\ref{sec:mainalg} for each interval the total space is $O(t)$. We answer a batch of $b$ predecessor queries according to the following two cases: 
\begin{enumerate}
    \item If $b \leq t$ we process each predecessor query in the batch by identifying the at most two intervals containing the answer and then querying these predecessor data structures. In total, this uses $O(b \log \log (u/t)) = O(b \log \log (u/b))$ time. 
    \item If $b > t$ we simply merge the sorted batch of queries with the input  set using $O(b+t) = O(b)$ time. 
\end{enumerate}
Since the batch is sorted we can also return the output in sorted order in $O(b)$ time. It follows that the running time is bounded by $O(b \log \log (u/b)$. Plugging into to the algorithm of Section~\ref{sec:mainalg}, we obtain a data structure that uses $O(m)$ space and supports computing a state-set transition $\delta(P, \alpha)$ in time $O(|P|\log \log \frac{m}{|P|} + |\delta(P, \alpha)|)$. This completes the proof of Theorem~\ref{thm:finiteautomaton}.

Next consider Theorem~\ref{thm:regexmatching}. Let $Q$ be a string of length $n$ and let $S_0, \ldots, S_n$ be the state-sets in the simulation of $R$ on $Q$. We implement the state-set transitions using Theorem~\ref{thm:finiteautomaton}. Note that each state-set transition produces the output in sorted order as required. Since logarithms are concave we have that the total time for the state-set transitions is 
\[
O\left(\sum_{i=0}^n |S_i| \log \log \frac{m}{|S_i|}\right) = O\left((n+1)\frac{\Delta}{n+1} \log \log \frac{m}{\Delta/(n+1)}\right) = O\left(\Delta \log \log \frac{nm}{\Delta}\right) \;.
\]
The algorithm uses $O(m)$ space to store the representation of $R$ and at most two state sets during the simulation. This completes the proof of Theorem~\ref{thm:regexmatching}.  


\section{Conditional Lower Bound}\label{sec:lowerbound}
We now prove the conditional lower bound of Theorem~\ref{thm:lowerbound}. Our lower bound follows the reduction of Backurs and Indyk~\cite{BI2016} from the orthogonal vectors problem (OVP) to regular expression matching.

The orthogonal vectors problems is defined as follows.  Given two sets $A,B \subseteq \{0,1\}^d$ such that $|A|=M$, $|B|=N$, determine if there exists $a\in A$ and $b\in B$ such that $a\cdot b = 0$. For any $M  = \Theta(N^\alpha)$ for some $\alpha \in (0,1]$ and any constant $\epsilon > 0$, any algorithm for OVP with running time $O((MN)^{1-\epsilon})$ violates SETH for $d = \omega(\log N)$~\cite{Williams05,BK15}.

Backurs and Indyk~\cite{BI2016} showed hardness of regular expression matching using a reduction from OVP. Given an instance of OVP they show how to construct a regular expression $R'$ and a string $Q'$ in $O(Nd)$ time such that $Q'$ matches $R'$ if and only if there exists  $a\in A$ and $b\in B$ such that $a\cdot b = 0$. The reduction works in $O(Nd)$ time, the lengths of both $R'$ and $Q'$ is $\Theta(Nd)$, and the alphabet is $\{x, y\}$. The regular expression $R'$ has the form
$$R' = \left (\bigodot _{j=1}^{|Q'|} (x^\ast y^\ast) \right)\cdot P \cdot \left (\bigodot _{j=1}^{|Q'|} (x^\ast y^\ast) \right)\;.$$
Here $P$ is a regular expression of length $O(Md)$ with the property that a substring of $Q'$ can be derived from $P$ if and only if there exists  $a\in A$ and $b\in B$ such that $a\cdot b = 0$. The precise definition of $R'$, $Q'$, and $P$ can be found in~\cite{BI2016}. 

We claim that $\Delta_{R',Q'} = \Theta(|Q'|^2)$. To see this first note that $\Delta_{R',Q'}$ is at most $|R'||Q'| + 1 = O(|Q'|^2)$. For the lower bound, consider the sequence $S_0, S_1, \ldots, S_{|Q'|}$ of state sets in the NFA simulation, and focus on the first $2|Q'|$ positions in $R'$, i.e., the positions corresponding to the subexpression immediately before $P$. Since $Q$ is a string of $x$s and $y$s, we have $|S_1| = |Q'|$ and $|S_i| \geq |S_{i-1}|-1$ . Thus, $\Delta_{R',Q'} = \Omega(|Q'|^2)$ and hence the claim follows. See Figure~\ref{fig:lowerbound}(a).

\begin{figure}
    \centering
    \includegraphics[scale=0.42]{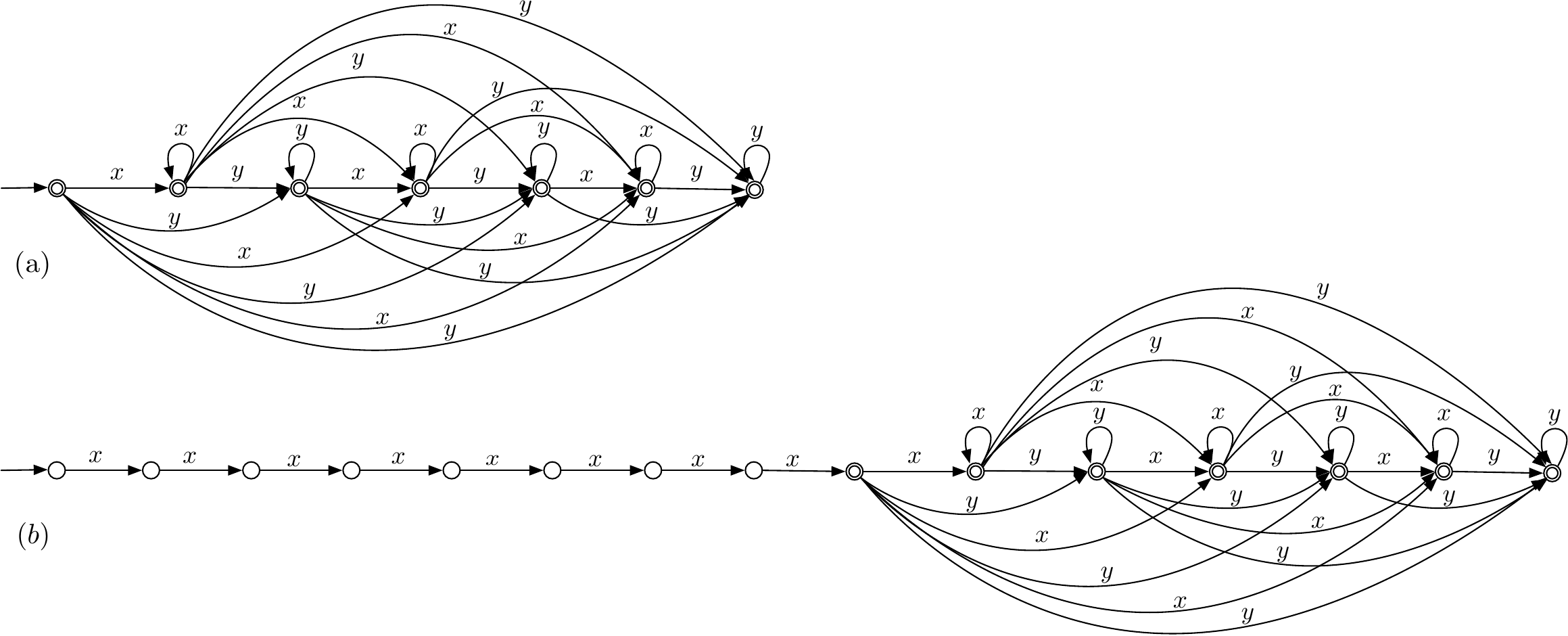}
    \caption{(a) The position automaton for $x^\ast y^\ast x^\ast y^\ast x^\ast y^\ast$. (b)  The position automaton for $x^8 x^\ast y^\ast x^\ast y^\ast x^\ast y^\ast$.}
    \label{fig:lowerbound}
\end{figure}

We can now prove the following theorem.
\begin{theorem}\label{thm:lbreduction}
    Given $A = \{a^1, \ldots, a^N\} \subseteq \{0,1\}^d$ and $B = \{b^1, \ldots, b^M\} \subseteq \{0,1\}^d $ and a constant $\gamma$, where $0 < \gamma \leq 1$, we can construct in  
    $O((Nd)^{2/(1+\gamma) })$
    time a regular expression $R$ and a string $Q$, such that there exists  $a \in A$ and $b\in B$ where $a\cdot b = 0$ if and only if $Q\in L(R)$. The size of $R$ and $Q$ is $\Theta((Nd)^{2/(1+\gamma) })$ and $\Delta_{R,Q} = \Theta(|Q|^{1+\gamma}) = \Theta(N^2d^2)$.
\end{theorem}

\begin{proof}
    We construct our instance $R$, $Q$ from $R'$ and $Q'$ as follows. 
    $$Q = x^{\ell}\cdot Q' \qquad \textrm{and} \qquad R = x^{\ell}\cdot R'$$
    where $\ell = (Nd)^{2/(1+\gamma)}$. 

Clearly, $Q$ matches $R$ if and only if $Q'$ matches $R'$. Furthermore, the NFA simulation on the first $\ell$ characters must produce singleton state sets  (see Figure~\ref{fig:lowerbound}(b)). Hence, we have that $\Delta_{R,Q} = \ell + \Delta_{R',Q'} = \Theta((Nd)^2)$. Since $|Q| = \ell + |Q'| = \Theta((Nd)^{2/(1+\gamma)})$ we have that $\Delta = \Theta(|Q|^{1+\gamma)})$.
\end{proof}

Theorem~\ref{thm:lowerbound} follows directly from Theorem~\ref{thm:lbreduction}, since an $O(\Delta^{1-\epsilon}) = O((Nd)^{2-2\epsilon})$ time algorithm for regular expression matching would imply a $O((Nd)^{2-2\epsilon}+ (Nd)^{2/(1+\gamma)})$ algorithm for OVP.


\section{Acknowledgments}
We thank the anonymous reviewers whose comments and suggestions signiﬁcantly improved the presentation of the paper.

\bibliography{paper}

\end{document}